\documentclass[11pt]{article}
\usepackage{authblk}
\title{Real numbers equally compressible in every base}
\author[1]{Satyadev Nandakumar}
\author[1]{Subin Pulari}
\affil[1]{
  Department of Computer Science and Engineering\\
  Indian Institute of Technology Kanpur,
  Kanpur, Uttar Pradesh, India.
}
\affil[]{\{\textit{satyadev,subinp}\}@cse.iitk.ac.in}

\usepackage{amsmath,amsthm,amssymb}
\usepackage[margin=1in]{geometry}
\usepackage[colorlinks=true]{hyperref}
\usepackage{array}
\usepackage{longtable}

\usepackage{amsfonts}
\usepackage{amsmath} 
\usepackage{amssymb}
\usepackage{enumerate}

\newcommand{\PROOF}{\begin{proof}}

\newcommand{\QED}{\end{proof}}

\newcommand{\N}{\mathbb{N}}
\newcommand{\Z}{\mathbb{Z}}
\newcommand{\Q}{\mathbb{Q}}
\newcommand{\R}{\mathbb{R}}

\renewcommand{\dim}{{\mathrm{dim}}}

\usepackage[svgnames]{xcolor}

\def\<{\left\langle}
\def\>{\right\rangle}

\theoremstyle{plain}
\newtheorem{theorem}{Theorem}[section]
\newtheorem{lemma}[theorem]{Lemma}
\newtheorem{corollary}[theorem]{Corollary}

\theoremstyle{definition}

\newtheorem{definition}[theorem]{Definition}

\begin{document}

\maketitle
\begin{abstract}
This work solves an open question in finite-state compressibility
posed by Lutz and Mayordomo \cite{LutzMayordomo2021} about
compressibility of real numbers in different bases.

Finite-state compressibility, or equivalently, finite-state dimension,
quantifies the asymptotic \emph{lower density of information} in an
infinite sequence.

Absolutely normal numbers, being finite-state incompressible in every
base of expansion, are precisely those numbers which have finite-state
dimension equal to $1$ in every base. At the other extreme, for
example, every rational number has finite-state dimension equal to $0$
in every base.

Generalizing this, Lutz and Mayordomo in \cite{LutzMayordomo2021}
posed the question: are there numbers which have absolute positive
finite-state dimension strictly between 0 and 1 - equivalently, is
there a real number $\xi$ and a compressibility ratio $s \in (0,1)$
such that for every base $b$, the compressibility ratio of the
base-$b$ expansion of $\xi$ is precisely $s$? It is conceivable that
there is no such number. Indeed, some works explore ``zero-one'' laws
for other feasible dimensions \cite{FortnowHitchcockPavanWang} -
\emph{i.e.}  sequences with certain properties either have feasible
dimension 0 or 1, taking no value strictly in between.

However, we answer the question of Lutz and Mayordomo affirmatively by
proving a more general result. We show that given any sequence of
rational numbers $\langle q_b \rangle_{b=2}^{\infty}$, we can
\emph{explicitly construct} a single number $\xi$ such that for any
base $b$, the finite-state dimension/compression ratio of $\xi$ in
base-$b$ is $q_b$. As a special case, this result implies the
existence of absolutely dimensioned numbers for \emph{any} given
rational dimension between $0$ and $1$, as posed by Lutz and
Mayordomo.

In our construction, we combine ideas from Wolfgang Schmidt's
construction of absolutely normal numbers from \cite{Schmidt62},
results regarding low discrepancy sequences and several new estimates
related to exponential sums.
\end{abstract}

\newpage
\tableofcontents

\newpage
\section{Introduction}
Finite-state compressibility is the lower asymptotic ratio of
compression achievable on an infinite string using
information-lossless finite-state compressors \cite{LZ78}
\cite{Sheinwald1994}. Finite-state dimension was originally defined by
Dai, Lathrop, Lutz and Mayordomo in \cite{dai2004finite} using
finite-state $s$-gales, as a finite-state analogue of Hausdorff
dimension \cite{LutzDimension2003}
\cite{athreya2007effective}. Surprisingly, these notions are
equivalent \cite{dai2004finite} \cite{bourke2005entropy}. They also
have several equivalent characterizations in terms of automatic
Kolmogorov complexity \cite{KozachinskiyAutomatic2021}, finite-state
predictors \cite{Hitchcock2003}, block-entropy rates
\cite{Hitchcock2003}, etc. establishing their mathematical robustness.

In this work, we solve an open question posed by Lutz and Mayordomo
\cite{LutzMayordomo2021} about the existence of numbers whose
finite-state compression ratio does not depend on the base of
expansion, where the compression ratio is neither 0 nor 1. We explain
the context behind this question below.

In \cite{SchnorrStimm72}, Schnorr and Stimm show the celebrated result
that finite-state incompressibility in base-$b$ is equivalent to Borel
normality in base-$b$, establishing a deep connection between
information theory and metric number theory. Absolutely normal numbers
\cite{Aistleitner2017} \cite{BecherHeiberSlaman2015}
\cite{LutzMayordomo2021} \cite{Schmidt60} \cite{Schmidt62}, being
finite-state incompressible in every base-$b$, are precisely the set
of numbers whose finite-state dimension is $1$ in every base. Almost
every number in $[0,1]$ is absolutely normal (see
\cite{Billingsley95}). Several explicit constructions absolutely
normal numbers are known \cite{Schmidt62} \cite{Aistleitner2017}
\cite{LutzMayordomo2021}. 

At the other extreme, there are numbers which have finite-state
dimension 0 in every base. For example, any rational number in $[0,1]$
has eventually periodic expansion in any base $b$ and hence,
finite-state dimension $0$ in any base of expansion.

But, no method for even proving the existence of numbers with absolute
finite-state dimension strictly between $0$ and $1$ are known at the
time of writing this paper. Lutz and Mayordomo in
\cite{LutzMayordomo2021} ask the following question:
\begin{quote}
	\emph{Does there exist a real number $\xi$ whose finite-state
        dimension/compressibility, $\dim^b_{FS}(\xi)$, does not depend
        on the choice of base $b$, and such that
        $0<\dim^b_{FS}(\xi)<1$?}
\end{quote}

In this work, we answer the question in the affirmative, by proving a
stronger result: \emph{Given any list $\langle q_b
\rangle_{b=1}^{\infty}$ of rationals in $(0,1]$ (respecting the
natural equivalence between bases $b$), there exists an explicitly
constructible number $\xi$ such that $\dim_{FS}^b(x)=q_b$ for any base
$b$}. In the special case when every $q_b$ is equal, this provides an
  affirmative answer (and an explicitly constructible example) to the
  above question posed by Lutz and Mayordomo.

We now state our main result. Two positive integer bases $r$ and $s$
are said to be equivalent, denoted $r \sim s$, if there are $m, n \in
\N$ such that $r^m = s^n$. (Formal definitions follow in section
\ref{sec:preliminaries}). If $r \sim s$, then we can verify that for
any $\xi \in [0,1]$, $\dim_{FS}^r(\xi)=\dim_{FS}^s(\xi)$. Our main
result is the following.
\begin{theorem}
\label{thm:maintheorem}
Let $\langle q_b \rangle_{b=1}^{\infty}$ be a sequence of rationals in
$(0,1]$ such that for any $r$ and $s$, if $r \sim s$, then $q_r =
  q_s$. Then, there exists a $\xi \in [0,1]$ such that for any base
  $b$, $\dim^b_{FS}(\xi)=q_b$.
\end{theorem}
When $q_b=q$ for every base $b$, for some $q \in \Q \cap [0,1]$, we
have the following corollary.
\begin{corollary}
 \label{cor:absolutelydimensionednumbers}
 Let $q \in [0,1] \cap \Q$. Then, there exists a $\xi \in [0,1]$ such
 that $\dim^{b}_{FS}(\xi)=q$ for every base $b \geq 2$. 
\end{corollary}

For $q \in (0,1]$, the existence of $\xi$ in the above corollary
  follows from Theorem \ref{thm:maintheorem}. For $q=0$, any rational
  number $\xi$ in $[0,1]$ satisfies the required
  conclusion. Therefore, this corollary provides a positive answer to
  the question posed by Lutz and Mayordomo in
  \cite{LutzMayordomo2021}. 

Explicit constructions of numbers with specific compressibility ratios
often use combinatorial techniques (for example, see
\cite{Cham33}). Combinatorial constructions are often simpler to
understand. However, in our work, we control multiple bases and
dimensions in each base. Since this implies working with different
alphabets simultaneously, a combinatorial approach to the solution is
not easy. We take an approach which involves exponential sums, which
allows us to handle the construction requirements successfully.

In this construction, we modify and combine techniques from Wolfgang
Schmidt's construction of absolutely normal numbers in
\cite{Schmidt62} along with results regarding low discrepancy
sequences \cite{GalGal64} \cite{Philipp74} and several new
estimates. Schmidt's method has been generalized to construct numbers
exhibiting specific kinds of normality (or non-normality) in different
bases (see for example \cite{BecherBugeaudSlaman2016}). Our
construction is a generalization of Schmidt's method which yields
numbers having prescribed rates of information in different bases.

The construction consists of multiple stages. In every stage we
arrange a \emph{controlled oscillation} of block entropies in a
particular base $b$ between $q_b$ and $1$, by fixing digits in base
$b$ using multiple \emph{low discrepancy} strings. Meanwhile we
\emph{stabilize} the entropies in other non-equivalent bases around
$1$, using generalizations of bounds from \cite{Schmidt62} and
estimates relating exponential averages to block entropies. We also
ensure that the block entropies in different bases do not \emph{fall
too much} during the transition between consecutive stages. This is
accomplished by using estimates involving low discrepancy strings and
lower bounds based on the concavity of the Shannon entropy function
(\cite{CovTho91}).

The rest of the paper is organized as follows. After the preliminaries
in section \ref{sec:preliminaries}, we outline the construction and
the \emph{requirements} that it should satisfy in section
\ref{sec:overview}. We also prove that if the requirements hold, then
Theorem \ref{thm:maintheorem} follows. In section
\ref{sec:technicallemmas}, we develop some important technical tools
required for the construction. In section \ref{sec:mainconstruction}
we describe our stage-wise construction in detail. Finally, in section
\ref{sec:verification}, we \emph{verify} that the construction in
section \ref{sec:mainconstruction} satisfies all the
\emph{requirements} given in section \ref{sec:overview}. For the convenience of the reader, we provide a table of notations and terminology in section \ref{sec:tableofnotationsandterminology}.

\section{Preliminaries}
\label{sec:preliminaries}
\subsection{Basic definitions and notation}
\label{subsec:basicdefinitions}
We use $\Sigma$ to denote any finite alphabet\label{text:sigmadefinition}. For any natural number
$b >1$, $\Sigma_b$ denotes the alphabet $\{0,1,2,\dots b-1\}$\label{text:sigmabdefinition}. Observe
that for any $b_1 \leq b_2$, we have $\Sigma_{b_1} \subseteq
\Sigma_{b_2}$.
For any finite alphabet $\Sigma$, we use $\Sigma^*$ to represent the
set of finite binary strings and $\Sigma^\infty$ to represent the set
of infinite sequences in alphabet $\Sigma$. We use capital letters
$X$, $Y$ to denote infinite strings in any finite alphabet
$\Sigma_b$. We use small letters $w$, $z$ to represent finite strings
in any finite alphabet $\Sigma_b$. For any $X=X_1X_2X_3\dots$ from
$\Sigma^\infty$, we use $X_1^n$\label{text:x1ndefinition} to denote the prefix $X_1X_2 \dots
X_n$ and for any $w \in \Sigma^*$, $w_1^n$\label{text:w1ndefinition} to represent $w_1w_2\dots
w_n$.

Greek letters $\alpha$, $\beta$ and $\gamma$ are used to represent
constants which are either absolute or dependent on some of the
parameters which are relevant in the context of their use. Small
letters $x$, $y$ and the Greek letter $\xi$ are used to denote general
real numbers in $[0,1]$. Small letters $b$, $r$ and $s$ are used for
representing integer bases of expansion greater than or equal to
$2$. Now, we define the occurrence (sliding) probabilities of finite
strings.

\begin{definition}
\label{def:occurrencecountandprobability}
 We define the \emph{occurrence count} of $z \in \Sigma^*$ in $w \in
 \Sigma^*$ denoted $N(z,w)$, as $N(z,w)=\lvert \{i \in [1,\lvert w
   \rvert-\lvert z \rvert+1] : w_i^{i+\lvert s \rvert-1} =
 z\}\rvert$. The \emph{occurrence probability} of $z$ in
 $w$ denoted $P(z,w)$, is defined as
 $P(z,w)=N(z,w)/(\lvert w \rvert-\lvert z \rvert+1)$.
\end{definition}

The following is the definition of the base-$b$ finite-state dimension
of a sequence $\xi$ using the block entropy characterization of
finite-state dimension given in \cite{bourke2005entropy}, which we use
in this work instead of the original definition using $s$-gales from
\cite{dai2004finite}.

\begin{definition}[\cite{dai2004finite}, \cite{bourke2005entropy}]
\label{def:finitestatedimension}
For a given base $b$ and a block length $l$, we define the \emph{$l$-length block
entropy} over $w \in \Sigma_b^*$ as follows.
\begin{align*}
H_l^b(w) = -(l \log(b))^{-1}\quad\sum\nolimits_{z \in \Sigma_b^l}\quad
P(z,w)\log P(z,w).
\end{align*}
Let $\xi \in [0,1]$ and let $X \in \Sigma_b^\infty$ represent the
base-$b$ expansion of $\xi$ (for numbers having two base-$b$
expansions, let $X$ denote any one of these). The base-$b$
\emph{finite-state dimension} of $\xi$, denoted $\dim^b_{FS}(\xi)$, is
defined by
\begin{align*}
\dim^b_{FS}(\xi) = \inf_{l} \liminf_{n \to \infty}
H_l^b(X_1^{n}).
\end{align*}
\end{definition}

Numbers of the form $k/b^n$ for some $k \in \Z$ and $n \in \N$ have
two base-$b$ expansions. But their finite-state dimension is equal to $0$
irrespective of the infinite sequence chosen as the base-$b$ expansion
of $x$. Therefore, the base-$b$ finite-state dimension is 
well-defined for every $\xi \in [0,1]$.

\textbf{Remark:} The fact that $\dim_{FS}(x)$ is equivalent to the lower finite-state compressibility of $x$ using lossless finite-state compressors, follows from the results in \cite{LZ78} and \cite{dai2004finite}. 

For every $x \in \R$, let $e(x)=e^{2\pi i x}$\label{text:exdefinition}. For $x \in \R$ and $d >
0$, let $B_d(x)$\label{text:bdxdefinition} denote the open neighborhood of $x$ having radius
$d$, i.e, $B_d(x)=\{y \in \R: \lvert y-x \rvert < d \}$. For any $w
\in \Sigma_b^*$, let $v_b(w)$\label{text:vbwdefinition} denote the rational number,$v_b(w)=
\sum_{i=1}^{\lvert w \rvert} w_i b^{-i}$. Therefore, the interval
$I^b_w=[v_b(w),v_b(w)+b^{-\lvert w \rvert})$\label{text:ibwdefinition} denotes the set of all
  numbers in $[0,1]$ whose base-$b$ expansion begins with the string
  $w$. The \emph{characteristic function} $\chi_w$\label{text:chiwdefinition} of a string $w \in
  \Sigma_b^*$ is defined as, $\chi_{w}(x)=1$ if $x \in I^b_w$ and $\chi_{w}(x)=0$
  otherwise.

 The following is a well-known Fourier series approximation for characteristic functions of cylinder sets (see \cite{Schmidt60}, \cite{Koksma74}) which acts as the basic \emph{bridge} between the combinatorial and analytic approaches.  
\begin{lemma}[\cite{Schmidt60}, \cite{Koksma74}]
\label{lem:koksmaapproximation}	
For any string $w \in \Sigma_b^*$ and $\delta>0$, there exists coefficients $C_t^i$ satisfying $\left\lvert C^i_t \right\rvert \leq \frac{2}{t^2 \delta}$, such that
\begin{align*}
b^{-\lvert w \rvert}-\delta + \sum\limits_{t \in
  \mathbb{Z}\setminus\{0\}} C_t^1 e(tx)\quad \leq\quad
\chi_{w}(x)\quad \leq\quad b^{-\lvert w \rvert}+\delta +
\sum\limits_{t \in \mathbb{Z}\setminus\{0\}} C_t^2 e(tx). 
\end{align*}
\end{lemma}

\subsection{Low discrepancy sequences}
\label{subsec:lowdiscrepancy}
Let $X_1X_2X_3 \dots$ be the sequence in $\Sigma_b^\infty$
representing the base-$b$ expansion of $x \in [0,1]$. For any real
interval $(\alpha_1, \alpha_2) \subseteq [0,1]$ let
$R^b(x,n,\alpha_1,\alpha_2)$\label{text:rbdefinition} be defined as follows.
\begin{align*}
R^b(x,n,\alpha_1,\alpha_2) = \left\lvert \frac{\lvert \{1 \leq i \leq n \mid 0.X_i X_{i+1} X_{i+2}\dots \in (\alpha_1,\alpha_2)\} \rvert}{n}	 - (\alpha_2-\alpha_1)\right\rvert
\end{align*}
The \emph{discrepancy function} $D^b_n(x)$\label{text:dbndefinition} is defined to be the
supremum over all $(\alpha_1,\alpha_2) \subseteq [0,1]$ of
$R^b(x,n,\alpha_1,\alpha_2)$ (\cite{GalGal64}, \cite{Philipp74}).  The
following theorem regarding the low discrepancy of almost every real
number follows from the results in \cite{GalGal64} and
\cite{Philipp74}.
\begin{theorem}[\cite{GalGal64},\cite{Philipp74}]
\label{thm:lowdiscrepancytheorem}
For any base $b$, there exists a constant $C_b$ such that for almost
every $x$, $\limsup_{n \to \infty} \frac{n D^b_n(x)}{\sqrt{n \log \log n}} < C_b$.
\end{theorem}
The following lemma is a corollary of Theorem \ref{thm:lowdiscrepancytheorem}.
\begin{lemma}
\label{lem:lowdiscrepancyfinitestrings}
 For any base $b$, there exists a constant $C_b$ such that for any
 $\epsilon >0$, there exists $N_b(\epsilon)$ such that outside a set
 of measure at most $\epsilon$, for any $\alpha_1<\alpha_2$ and $n
 \geq N_b(\epsilon)$, we have $R^b(x,n,\alpha_1,\alpha_2)$ is strictly
 less than $C_b \cdot \frac{\sqrt{\log \log n}}{\sqrt{n}}$.
\end{lemma}
\begin{proof}
	Theorem \ref{thm:lowdiscrepancytheorem} implies that for almost every $x$, there exists a minimum number $N_x$ such that,
	\begin{align*}
		 R^b(x,n,\alpha_1,\alpha_2) < C_b \cdot \frac{\sqrt{\log \log N}}{\sqrt{n}}
	\end{align*}
	for any $\alpha_1<\alpha_2$ and every $n \geq N_x$. For every $i \in \N$, define $U_i=\{x : N_x \geq i\}$. $U_i$'s are a monotonically decreasing sequence of sets such that $\mu(\bigcap_{i \in \N} U_i)=0$. Using the continuity of measure from above (see \cite{Billingsley95}), there exists a large enough $N_b(\epsilon)$ such that $\mu(U_i) < \epsilon$ for every $i \geq N_b(\epsilon)$. This completes the proof of the lemma.
\end{proof}
The following is a corollary of the above lemma that we need in our construction.
\begin{corollary}
\label{cor:lowdiscrepancy}
 For any base $b$, there exists a constant $C_b$ such that for any $\epsilon >0$,  there exists $N_b(\epsilon)$ satisfying the following: for any $n'> N_b(\epsilon)$, every string $w$ of length $n'$ except at most $\epsilon \cdot b^{n'}$ of them is such that given any string $z$ of length at most $n'-N_b(\epsilon)$ and $n$ ranging from $N_b(\epsilon)$ to $n'-\lvert z\rvert$, we have
 \begin{align}
 \label{eq:lowdiscrepancycondition}
  \left\lvert \frac{N(z,w_1^n)}{n}-\frac{1}{b^{\lvert z \rvert}} \right\rvert < C_b \cdot  \frac{\sqrt{\log \log n}}{\sqrt{n}}.
 \end{align}
\end{corollary}

For any base $b$ and $\epsilon = \frac{1}{2}$, let the collection of all strings of length $n'$ satisfying inequality (\ref{eq:lowdiscrepancycondition}) be referred to as $\mathcal{G}_b^{n'}$. Corollary \ref{cor:lowdiscrepancy} therefore says that for any length $n'> N_b(1/2)$, $\vert \mathcal{G}_b^{n'}\rvert$ is at least $b^{n'}/2$.

\begin{proof}[Proof of Corollary \ref{cor:lowdiscrepancy}]
	$I^b_w=[v_b(w),v_b(w)+b^{-\lvert w \rvert})$ represents the interval containing exactly those real numbers in $[0,1]$ whose base-$b$ expansion starts with the string $w$. Setting $\alpha_1=v_b(w)$ and $\alpha_2=v_b(w)+b^{-\lvert w \rvert}$, it follows from Lemma \ref{lem:lowdiscrepancyfinitestrings} that for $x=0.X_1 X_2 X_3 \dots$ outside a set of measure at most $\epsilon$,
	\begin{align}
	\label{eq:lowdiscrepancy1}
	\left\lvert \frac{N(z,X_1^n)}{n}-\frac{1}{b^{\lvert z \rvert}} \right\rvert	< C_b \frac{\sqrt{\log \log n}}{\sqrt{n}}
	\end{align}
	for $n' \geq N_b(\epsilon)$, string $z$ of length less than $n'-N_b(\epsilon)$ and $n$ ranging from $N_b(\epsilon)$ to $n'-\lvert z \rvert$. This implies that the number of strings of length $n' \geq N_b(\epsilon)$ such that,
	\begin{align}
	\label{eq:lowdiscrepancy2}
	\left\lvert \frac{N(z,X_1^n)}{n}-\frac{1}{b^{\lvert z \rvert}} \right\rvert	\geq C_b \frac{\sqrt{\log \log n}}{\sqrt{n}}
	\end{align}
	for some string $z$ of length less than $n'-N_b(\epsilon)$ and some $n$ between $N_b(\epsilon)$ to $n'-\lvert z \rvert$ is at most $\epsilon \cdot b^n$. If this is not the case, then the union of the cylinder sets corresponding to the strings violating (\ref{eq:lowdiscrepancy1}) is a set of measure greater than $\epsilon$ in which each element violates (\ref{eq:lowdiscrepancy1}). Since, this leads to a contradiction, the number of strings of length $n' \geq N_b(\epsilon)$ such that (\ref{eq:lowdiscrepancy2}) holds, is at most $ \epsilon \cdot b^n$. This completes the proof of the corollary.
	\end{proof}

\subsection{Schmidt's construction method (\cite{Schmidt62})}
\label{subsec:schmidtsmethod}
The individual \emph{steps} of the construction in the proof of Theorem \ref{thm:maintheorem} are based on the construction method used by Schmidt in his construction of absolutely normal numbers \cite{Schmidt62}. As far as possible we use Schmidt's notation from \cite{Schmidt62} in this paper. We give a brief account of Schmidt's method below. 

Let $u(1),u(2),u(3),\dots$\label{text:umdefinition} be any sequence of natural
numbers. This sequence represents the bases in which we \emph{fix} the
digits in each step of our construction. Now, as in \cite{Schmidt62},
define, $\langle m \rangle = \lceil e^{\sqrt{m}}+2 \cdot u(1) \cdot
m^3 \rceil$\label{text:anglemdefinition} and $\langle m; r \rangle = \lceil \langle m \rangle /
\log(r)\rceil$\label{text:anglemrdefinition}. Also define, $a_m= \langle m ; u(m) \rangle$ and $b_m=
\langle m+1 ; u(m) \rangle$. Notice that for any $m$, if
$u(m)=u(m+1)$\label{text:amdefinition}, we have $b_m=a_{m+1}$\label{text:bmdefinition}. For convenience of notation, we
define $\langle 0 \rangle =0$. Let $p$\label{text:pdefinition1} be a function from $\N$ to $\N$
such that for every $m \geq 1$, $p(u(m)) \leq u(m)$. The exact
function $p$ we use in our construction is defined in section
\ref{sec:overview}. For any $m \geq 1$ and positive real
number $\lambda$, let $g_m(\lambda)$ denote the smallest natural
number such that, $g_m(\lambda) u(m)^{-a_m} \geq \lambda$. Now, define
$\eta_m(\lambda)=g_m(\lambda) u(m)^{-a_m}$. We define
$\sigma_m(\lambda)$\label{text:sigmamlambdadefinition} to be the set of numbers,
\begin{align*}
  \eta_m(\lambda)+
  c^{u(m)}_{a_m+1} u(m)^{-(a_m+1)}+
  c^{u(m)}_{a_m+2} u(m)^{-(a_m+2)}+
  \dots +
  c^{u(m)}_{b_m-2} u(m)^{-(b_m-2)}
\end{align*}
with coefficients $c^{u(m)}_i$ taking values from the set $\{0,1,\dots u(m)-1\}$. We define $\sigma_m^*(\lambda)$\label{text:sigmastarmlambdadefinition} to be set in which the coefficients $c^{u(m)}_i$ takes values from $\{0,1,\dots p(u(m))-1\}$. Let $\xi_0=0$ and let $\xi_1,\xi_2,\xi_3$ be a sequence of real numbers such that, $\xi_m \in \sigma_m(\xi_{m-1})$ or  $\xi_m \in \sigma^*_m(\xi_{m-1})$\label{text:ximdefinition}. 

Since $\xi_m \geq \xi_{m-1}$, it follows that there exists $\xi \in [0,1]$ such that $\lim_{m \to \infty} \xi_m = \xi$. Furthermore, digits $c^{u(m)}_{a_m+1} c^{u(m)}_{a_m+2}\dots c^{u(m)}_{b_m-2}$ appears in positions $a_m+1$ to $b_m-2$ in the base-$u(m)$ expansion of $\xi$. This follows as a consequence of the following lemma.

\begin{lemma}[\cite{Schmidt62}]
\label{lem:digitsinlimitnumber}
$
\left\lvert \xi-\xi_m \right\rvert  < u(m)^{-(b_m - 2)}.
$
\end{lemma}
\begin{proof}
	For every $m$, $\xi_m \geq \xi_{m-1}$ and thus we have,
\begin{align*}
\left\lvert \xi_m-\xi_{m-1} \right\rvert &= \xi_m-\xi_{m-1}\\
&=  \xi_m-\eta_m(\xi_{m-1})  +  \eta_m(\xi_{m-1}) - \xi_{m-1} \\
&\leq \frac{1}{u(m)^{a_m}} + \frac{1}{u(m)^{a_m}}\\
&= \frac{2}{u(m)^{a_m}}.
\end{align*}
Then, 
\begin{align*}
	\sum\limits_{i=m+1}^{\infty} \frac{1}{u(u)^{a_i}} &\leq \frac{1}{e^{\langle m+1 \rangle}} (1+\frac{1}{e^{2}}+\frac{1}{e^{4}}+\dots)\\
		&< \frac{3}{2} \frac{1}{e^{\langle m+1 \rangle}}\\
		&< \frac{3}{2} \frac{1}{u(m)^{b_m - 1}}\\
		&\leq \frac{1}{2} \frac{1}{u(m)^{b_m - 2}}.
\end{align*}
Therefore, the limit $\xi$ satisfies,
\begin{align*}
0\leq \xi-\xi_m =\left\lvert \xi-\xi_m \right\rvert  < \frac{1}{u(m)^{b_m - 2}}.	
\end{align*}
\end{proof}

Adapting the technique in \cite{Schmidt62}, we use the following function $A_m$ while choosing $\xi_m$ from $\sigma_m(\xi_{m-1})$ or $\sigma^*_m(\xi_{m-1})$ in the construction of the required number $\xi$,
\begin{align}
\label{eq:amxdefinition}
 A_m(x) &= \sum\nolimits_{\substack{t=-m\\t \neq 0}}^{m}\quad \sum\nolimits_{\substack{h=1 \\ u(h) \not\sim u(m)}}^{m}  \left\lvert \sum\nolimits_{j=\langle m; u(h) \rangle+1}^{\langle m+1 ; u(h) \rangle} e(u(h)^{j-1} t x) \right\rvert^2.
\end{align}
\label{text:criteria}In our construction in each step, $\xi_m$ is chosen according to either of the following criteria:
\begin{enumerate}
\item\label{item:criterion1} (Criterion 1.) $\xi_m$ is any element of
  $\sigma_m^*(\xi_{m-1})$ such that $c^{u(m)}_{a_m+1}
  c^{u(m)}_{a_m+2}\dots c^{u(m)}_{b_m-2} \in
  \mathcal{G}_{p(u(m))}^{b_m-a_m+2}$ with the minimum $A_m$
  value among all elements of $\sigma^*_m(\xi_{m-1})$ satisfying this
  condition.
\item\label{item:criterion2} (Criterion 2.) $\xi_m$ is any element of
  $\sigma_m(\xi_{m-1})$ such that $ c^{u(m)}_{a_m+1}
  c^{u(m)}_{a_m+2}\dots c^{u(m)}_{b_m-2} \in
  \mathcal{G}_{u(m)}^{b_m-a_m+2}$ with the minimum $A_m$
  value among elements of $\sigma_m(\xi_{m-1})$ satisfying this
  condition.
\end{enumerate}
\section{Overview of the Proof of Theorem \ref{thm:maintheorem}}
\label{sec:overview}
We need to construct a number $\xi$ having finite-state dimension
equal to $q_b$ in base $b$ for every $b$. For every $b \geq 2$, let
$e_b$ and $d_b$ be natural numbers such that $q_b= e_b/d_b$\label{text:ebdbdefinition} in the
lowest terms. Below we demonstrate the construction of a number $\xi$
with dimension $q_b$ in base $b^{d_b}$ for every $b$. This number has
dimension equal to $q_b$ in base $b$ for every $b \geq 2$.

Let $\langle r_k \rangle_{k=1}^{\infty}$\label{text:rkdefinition} be any sequence of natural numbers greater
than or equal to $2$ such that every equivalence class of numbers (according to the relation $\sim$) has a unique representative in the sequence, which appears infinitely often. Furthermore, we require that $r_1=2$ and no consecutive elements in the sequence are equal.

It is straightforward to construct sequences satisfying the above
conditions.  Given the sequence $\langle q_b \rangle_{b=1}^{\infty}$,
define the function $p: \N \to \N$ as $p(b)=\lfloor b^{q_b} \rfloor$\label{text:pdefinition2}.
For every $k \geq 1$, define $v(k)=r_k^{d_{r_k}}$\label{text:vkdefinition} and
$v^*(k)=p(v(k))$\label{text:vstarkdefinition}.  From the defining property of the sequence $\langle
q_b \rangle_{b=1}^{\infty}$, we get that for any $k$, $
q_{r_k^{d_{r_k}}}=q_{r_k}$.  Therefore, $v^*(k) = r_k^{e_{r_k}}$.

Let the sequence $\langle u(m)\rangle$ be initially empty and let
$\xi_0=0$.  At every \emph{step} $m \geq 1$ in the construction, we
fix the $m$\textsuperscript{th} value of the sequence $\langle
u(m)\rangle$ and choose $\xi_m$ from $\sigma_{m}(\xi_{m-1})$ or
$\sigma^*_{m}(\xi_{m-1})$. The whole construction is divided into a
sequence of \emph{stages} such that each individual stage comprises of
two different \emph{substages}. Each of the \emph{substages} consist
of multiple consecutive \emph{steps}.

Suppose by stage $k-1$, $u(1)$, $u(2)$, $\dots$, $u(n_{k-1})$ have
been determined.  Then in the $k$\textsuperscript{th} stage, we set
$u(n_{k-1}+1) = u(n_{k-1}+2) = \dots = u(n_k) = v(k)$, and fix the
digits in the base $v(k)$ expansion of $\xi$. Hence,
$u(1)=v(1)=r_1^{d_{r_1}}=2^{d_{2}}$. Within the first substage, at
step $m$, we choose $\xi_m$ from $\sigma^*_m(\xi_{m-1})$ according to
Criterion \ref{item:criterion1}. Within the second substage, at step
$m$, we choose $\xi_m$ from $\sigma_m(\xi_{m-1})$ according to
Criterion \ref{item:criterion2}.

For any $k \geq 1$, let $X(k)$\label{text:xkdefinition} denote the infinite sequence in alphabet $\Sigma_{v(k)}$ representing the base-$v(k)$ expansion of $\xi$. Let $P_k^1$\label{text:pk1definition} denote the final step number contained within the first substage of stage $k$. Similarly, let $P_k^2$\label{text:pk2definition} denote the final step number contained within the second substage of stage $k$. For convenience, let $P_0^1=P_0^2=0$. Define $I_1^1=1$ and $I_{k}^1=\langle P_{k-1}^2+1; v(k) \rangle+1$\label{text:ik1definition} for $k>1$. Now, $I_{k}^1$ denotes the index of the initial digit in $X(k)$ fixed during stage $k$. Also, $F_k^1=\langle P_k^1+1;v(k) \rangle$\label{text:fk1definition} denotes the index of the final digit in $X(k)$ fixed during the first substage of stage $k$. Let $I_k^2=\langle P_k^1+1;v(k) \rangle+1$\label{text:ik2definition} denote the index of the initial digit in $X(k)$ fixed during the second substage of stage $k$. Finally, let $F_k^2=\langle P_k^2+1;v(k) \rangle$\label{text:fk2definition} denote the index of the final digit in $X(k)$ fixed during the second substage of stage $k$.  

\label{text:requirementsbegin}The lengths of the stages and substages ensure that the constructed
$\xi$ satisfies the following \emph{requirements}. For every $k \geq
1$, we have the following \emph{end of substage} requirements:

\begin{enumerate}
	\item $\mathcal{F}_k$\label{text:fkdefinition} : $\lvert H_l^{v(k)}(X(k)_1^n) -q_{r_k} \rvert \leq 2^{-k}$ for every $l \leq k$ when $n=F_k^1$. 
	\item $\mathcal{S}_{k,1}$\label{text:sk1definition} : $\lvert H_l^{v(k)}(X(k)_1^n) -1 \rvert \leq 2^{-(k+1)}$ for every $l \leq k$ when $n=F_k^2$. 
	\item $\mathcal{S}_{k,2}$\label{text:sk2definition} : If there exists $k'<k$ such that $v(k')=v(k+1)$, then $\lvert H_l^{v(k+1)}(X(k+1)_1^n) -1 \rvert \leq 2^{-k}$ for every $l \leq k$ when $n=\langle P_k^2+1;v(k+1) \rangle$. 
\end{enumerate}

Requirement $\mathcal{F}_k$ ensures that the block entropies of $\xi$
in base-$v(k)$ are \emph{close} to $q_{r_k}$ by the end of the first
substage of stage $k$. Similarly, $\mathcal{S}_{k,1}$ ensures that $
H_l^{v(k)}$ are \emph{close} to $1$ by the end of the second substage
of stage $k$. $\mathcal{S}_{k,2}$ ensures that the block entropies of
$\xi$ in base-$v(k+1)$ are \emph{close} to $1$ before the start of
stage $k+1$ (provided that $v(k+1)$ has appeared as $v(k')$ for some
$k'<k$).

For every $k>1$, the following requirement specifies the \emph{stability of non-equivalent base entropies}:
\begin{enumerate}
  \setcounter{enumi}{3}
\item $\mathcal{R}_k$\label{text:requirementrkdefinition} : For any $k' < k$ such that $v(k') \not\sim
  v(k)$, $\lvert H_l^{v(k')}(X(k')_1^n) -1 \rvert \leq 2^{-(k'+1)}$
  for every $l \leq k'$ when $\langle P_{k-1}^2+1; v(k') \rangle+1
  \leq n \leq \langle P_{k}^2+1; v(k') \rangle$. 
\end{enumerate}
Requirement $\mathcal{R}_k$ ensures that the block entropies of $\xi$ in any base $v(k')$ for $k'<k$ which is not equivalent to $v(k)$ remains \emph{stable} around $1$ \emph{throughout the course} of stage $k$.

Two particularly important requirements we need to enforce for every $k$ are the \emph{transition requirements}:
\begin{enumerate}
  \setcounter{enumi}{4}
	\item $\mathcal{T}_{k,1}$\label{text:tk1definition} : $H_l^{v(k)}(X(k)_1^n) \geq q_{r_k}-2^{-(k-1)}$ for every $l \leq k$ when $F_k^1 \leq n \leq F_k^2$.	
	\item  $\mathcal{T}_{k,2}$\label{text:tk2definition} : If there exists $k'<k$ such that $v(k')=v(k+1)$, then $H_l^{v(k+1)}(X(k+1)_1^n) \geq q_{r_{k+1}}-2^{-(k-1)}$ for every $l \leq k$ when $I_{k+1}^1 \leq n \leq F_{k+1}^1$.
\end{enumerate}
$\mathcal{T}_{k,1}$ ensures that the base-$v(k)$ entropies do not
\emph{fall much below} $q_{r_k}$ during the transition between the
first and second substage of stage $k$. Similarly, $\mathcal{T}_{k,2}$
ensures that the base-$v(k+1)$ entropies do not \emph{fall much below}
$q_{r_{k+1}}$ during the transition between stages $k$ and $k+1$.

We now show that if we satisfy the above requirements, then Theorem
\ref{thm:maintheorem} follows (It will then suffice to show that
these requirements are met by our construction).

\begin{proof}[Proof of Theorem \ref{thm:maintheorem}]
Assume that the construction satisfies the requirements $\mathcal{F}_k$, $\mathcal{S}_{k,1}$, $\mathcal{S}_{k,2}$, $\mathcal{T}_{k,1}$ and $\mathcal{T}_{k,2}$ for every $k \geq 1$ and $\mathcal{R}_k$ for $k > 1$. Let $b$ be an arbitrary base of expansion. Let $\bar{k}$ be the smallest number such that $r_{\bar{k}} \sim b$. Such a \emph{representative} $r_{\bar{k}}$ exists for any $b$ due to the condition imposed on the sequence $\langle r_k \rangle_{k=1}^{\infty}$ at the start of this section. Let $X (\bar{k})\in \Sigma^\infty_{v(\bar{k})}$ denote the base-$v(\bar{k})$ expansion of $\xi$. In order to prove Theorem \ref{thm:maintheorem}, it is enough to show that
$\inf_{l} \liminf_{n \to \infty} H_l^{v(\bar{k})}(X(\bar{k})_1^n)  = q_{r_{\bar{k}}}$.
This implies that,
$\dim_{FS}^b(\xi)=\dim_{FS}^{v(\bar{k})}(\xi)=\inf_{l} \liminf_{n \to
  \infty} H_l^{v(\bar{k})}(X(\bar{k})_1^n) =
q_{r_{\bar{k}}}=q_{b}$. Here we used the fact that the sequence of
rational dimensions is such that $q_r=q_s$ if $r \sim s$.  We show that for every $l \geq 1$, $\liminf_{n
  \to \infty} H_l^{v(\bar{k})}(X(\bar{k})_1^{n}) = q_{r_{\bar{k}}}$.

Fix any length $l$. Let $k'$ be any large enough number such that $k'
> \max\{\bar{k},l\}$ and $r_{k'} = r_{\bar{k}}$ (therefore
$v(k')=v(\bar{k})$). In the rest of the argument by referring to the index in
$X(\bar{k})$ at the \emph{start of stage $k$} we mean the index $n=
\langle P_{k-1}^2+1; v(\bar{k}) \rangle+1$, and the index at the
\emph{end of stage $k$} refers to the index $n= \langle P_{k}^2;
v(\bar{k}) \rangle$.  

Since the requirement $\mathcal{F}_{k'}$ is met, after the first
substage of stage $k'$, $H_l^{v(\bar{k})}=H_l^{v(k')}$ is inside
$B_{2^{-k'}}(q_{r_{\bar{k}}})=B_{2^{-k'}}(q_{r_{k'}})$. Since the
requirement $\mathcal{S}_{k',1}$ is met, by the end of stage $k'$,
$H_l^{v(\bar{k})}$ moves to $B_{2^{-k'}}(1)$ such that at any index
$n$ during this transition, $H_l^{v(\bar{k})}(X(\bar{k})_1^n) \geq
q_{r_{\bar{k}}} - 2^{-(k'-1)}$. This inequality follows from the fact
that $\mathcal{T}_{k',1}$ is satisfied. We know that $\langle r_k
\rangle_{k=1}^{\infty}$, satisfies $r_{k'} \neq r_{k'+1}$ (and therefore $v(k') \neq v(k'+1)$). Since
$\mathcal{R}_{k'+1}$ is satisfied, during the transition to the next
stage and during the course of the next stage, $H_l^{v(\bar{k})}$
remains inside $B_{2^{-(k'+1)}}(1)$. Furthermore, $H_l^{v(\bar{k})}$
remains inside $B_{2^{-(k'+1)}}(1)$ until stage $k''$ where $k''$ is
the smallest number such that $k''>k'$ and
$r_{k''}=r_{k'}=r_{\bar{k}}$. This follows from the fact that
$\mathcal{R}_{k'+i}$ is satisfied for every $i < k''-k'$.  Since,
$\mathcal{S}_{k''-1,2}$ is satisfied and $v(\bar{k})=v(k') = v(k'') =
v(k''-1+1)$, by the end of the second substage of stage $k''-1$,
$H_l^{v(\bar{k})}$ is inside $B_{2^{-k''}}(1)$. During stage $k''$,
$H_l^{v(\bar{k})}$ starts being inside $B_{2^{-(k'')}}(1)$ and moves
to $B_{2^{-k''}}(q_{r_{\bar{k}}})$ (since $\mathcal{F}_{k''}$ is
met). During this transition, since $\mathcal{T}_{k''-1,2}$ is met and
$v(k'')=v(k''-1+1)=v(k')$ for $k' < k''-1$, it follows that
$H_l^{v(\bar{k})}(X(k)_1^n) > q_{r_{\bar{k}}} - 2^{-(k''-2)}$. Therefore, $H_l^{v(\bar{k})}$ remains above $q_{r_{\bar{k}}} -
2^{-(k''-2)}$. Since $k'$ was arbitrary, the above observations together imply that, $\liminf_{n
  \to \infty} H_l^{v(\bar{k})}(X(k)_1^{n}) = q_{r_{\bar{k}}}$. This
completes the proof of Theorem \ref{thm:maintheorem}.
\end{proof} 

Hence, the proof of Theorem \ref{thm:maintheorem} is complete if we show the construction of a number $\xi$ satisfying all the above requirements. We demonstrate the construction of $\xi$ and verify that all the requirements are satisfied, in the following sections.

\section{Technical Lemmas for the Main Construction}
\label{sec:technicallemmas}

We require two main technical lemmas for the main construction in the proof of Theorem \ref{thm:maintheorem}. In order to state the first lemma, we require the following generalization of Lemma 5 from \cite{Schmidt62}. 
\begin{lemma}
 \label{lem:sinbound}
 Consider any two bases $r$ and $s$. Let $K$ and $l$ be natural numbers such that $\ell \geq s^K$. Then there exists a constant $\alpha(r,s)$ depending only on $r$ and $s$ such that,
 \begin{align*}
  \sum\limits_{n=0}^{N-1} \prod\limits_{i=K+1}^{\infty} \left\lvert \frac{\sin(p(s)\pi r^n \ell/s^i)}{p(s)\sin(\pi r^n \ell/s^i)} \right\rvert \leq 2\cdot N^{1-\alpha(r,s)}.
 \end{align*}
\end{lemma}

\begin{proof}
 The function,
 \begin{align*}
  f(x)=\left\lvert\frac{\sin(p(s)\pi x)}{p(s)\sin(\pi x)}\right\rvert
 \end{align*}
    has the limit $1$ as $x \to 0$ and $f(x)$ takes values strictly less then $1$ when $\lvert x \rvert < 1$. The first property follows from the fact that $\lim\limits_{x \to 0} \sin(x)/x = 1$ and the second fact follows from the inequality $\lvert \sin(nx) \rvert < n \lvert \sin(x) \rvert$, which is easily proved using induction on $n$ when $\lvert x \rvert < \pi$. If $x$ has an obedient digit pair \cite{Schmidt62} (i.e, a pair of digits in base $s$ that are not both equal to $0$ or both equal to $s-1$) then,
    \begin{align*}
     \left\lvert\frac{\sin(p(s)\pi x/s^i)}{p(s)\sin(\pi x/s^i)}\right\rvert &\leq \left\lvert\frac{\sin(p(s)\pi/s^2)}{p(s)\sin(\pi/s^2)}\right\rvert < \gamma_1 < 1.
    \end{align*}
    The constant $\gamma_1$ above depends only on $s$ (since the function $p$ is fixed). 
    Let $z_K(x)$ denote the number of obedient digit pairs $c_{i+1}c_i$ in $x$ for $i \geq K$ (where $c_{\lceil \log_s(x) \rceil} \dots c_2 c_1$ is the representation of $x$ in base $s$).  Lemma 4 from \cite{Schmidt62} implies that if $l \geq s^K$, then among the numbers $\ell, \ell r, \ell r^2 \dots \ell r^{N-1} $, there are at most $N^{1-a_{14}}$ numbers such that $z_K$ is smaller than $a_{15} \log N$ (where $a_{14}$ and $a_{15}$ are constants depending only on $r$ and $s$ and independent of $\ell$ and $K$). If for some $n$, $\ell r^n$ has $z_K$ greater than $a_{15} \log N$, then,
    \begin{align*}
    	\prod\limits_{i=K+1}^{\infty} \left\lvert \frac{\sin(p(s)\pi r^n \ell/s^i)}{p(s)\sin(\pi r^n \ell/s^i)} \right\rvert \leq \gamma_1^{a_{15}\log N} = N^{1-\gamma_2}
    \end{align*}
 	for some $\gamma_2 \in (0,1)$ dependent only on $r$ and $s$. Now, using the above bound along with Lemma 4 from \cite{Schmidt62}, we get that,
 	\begin{align*}
 		\sum\limits_{n=0}^{N-1} \prod\limits_{i=K+1}^{\infty} \left\lvert \frac{\sin(p(s)\pi r^n \ell/s^i)}{p(s)\sin(\pi r^n \ell/s^i)} \right\rvert \leq N^{1-a_{14}} + N^{1-\gamma_2} \leq 2 \cdot N^{1-\alpha(r,s)}
 	\end{align*}
	for some $\alpha(r,s) \in (0,1)$ dependent only on $r$ and $s$.
\end{proof}

Lemma 5 from \cite{Schmidt62} is a special case of the above lemma when $p(s)=2$. We assume that for any $r$ and $s$, the constant $\alpha(r,s)$ in the above lemma is at most $1/2$ and $\alpha(r,s)=\alpha(s,r)$. Now, we define the notion of good sequences of natural numbers.

\begin{definition}[\emph{Good sequences of natural numbers}]
\label{def:goodsequence}
	A sequence $\langle u(m) \rangle_{m=1}^{\infty}$ is a \emph{good sequence} of natural numbers if the following conditions are satisfied:
	\begin{enumerate}
		\item\label{item:goodsequence1}  $\prod\limits_{i=b_m-1}^{\infty} \left\lvert \frac{\sin(p(u(m))\pi /2^{i+1})}{p(u(m)) \sin(\pi /2^{i+1})} \right\rvert \geq \prod\limits_{i=1}^{\infty} \left\lvert \cos( \pi /2^{i+1}) \right\rvert$ for every $m > 1$.
		\item\label{item:goodsequence2} $\beta_m \geq \beta_1 \frac{1}{\sqrt[4]{m}}$ for every $m \geq 1$ where
		\begin{align*}
			\beta_m = \min\left(\{\alpha(u(i),u(j)):1\leq i \leq j \leq m \text{ such that } u(i) \not\sim u(j) \}\cup\left\{1/2\right\}\right).
		\end{align*}
		\item\label{item:goodsequence3} $u(m) \leq u(1) m$ for every $m \geq 1$.
		\item\label{item:goodsequence4} For any $m \geq 1$, if there exists any $m'<m$ such that $u(m') \neq u(m)$, then $b_m-a_m \geq \max\{N_{u(m)}(1/2),N_{p(u(m))}(1/2)\}$, where the constants on the right are from Corollary \ref{cor:lowdiscrepancy}. 
	\end{enumerate}
\end{definition}

From condition \ref{item:goodsequence2} and the fact that there does not exist any $u(j) \not\sim u(1)$ with $j \leq 1$, it follows that $\beta_1=1/2$. For every $i$, we use the notation $\beta'_m$\label{text:betaidefinition} to denote $\beta_m/2$. The existence of good sequences follows from following lemma and the fact that $b_m$ is increasing in $m$.
\begin{lemma}
\label{lem:convergenceofinfiniteproduct}
For any $n \in \N$, the infinite product 
$\prod\limits_{i=1}^{\infty} \left\lvert \frac{\sin(n\pi /2^{i+1})}{n \sin(\pi /2^{i+1})} \right\rvert$ is convergent.
\end{lemma}
In order to prove Lemma \ref{lem:convergenceofinfiniteproduct}, we need the following inequality.
\begin{lemma}
\label{lem:sinnxlowerbound}
For every $n \geq 2$ and $x$ with $\lvert x \rvert < 1$,
	\begin{align*}
 \frac{\sin(nx)}{n\sin(x)} \geq 1-\frac{(n^2-1)x^2}{6}.
\end{align*}
\end{lemma}
\begin{proof}
	We prove the statement using induction on $n$. Consider the base case when $n=2$. We have,
	\begin{align*}
	\frac{\sin(2x)}{2\sin(x)} = \cos(x) \geq 1-\frac{x^2}{2}=1-\frac{(2^2-1)x^2}{6}.	
	\end{align*}
	The inequality $\cos(x) \geq 1-x^2/2$ follows easily from the Taylor series expansion of $\cos(x)$ since $\lvert x \rvert <1$.
	
	Assume that the statement in the conclusion holds for arbitrary $n$. We now show that this implies the required conclusion for $n+1$. We have,
	\begin{align*}
	\frac{\sin((n+1)x)}{(n+1)\sin(x)} &= \frac{\sin(nx)\cos(x)+\cos(nx)\sin(x)}{(n+1)\sin(x)}\\
	&= \frac{\sin(nx)}{n\sin(x)}\cdot \frac{n}{n+1}\cdot \cos(x)+\frac{1}{n+1}\cdot \cos(nx).
	\end{align*}
	Using the induction hypothesis and the inequality $\cos(x) \geq 1-x^2/2$, we get,
	\begin{align*}
	\frac{\sin((n+1)x)}{(n+1)\sin(x)} &\geq \left( 1-\frac{(n^2-1)x^2}{6} \right)\cdot \frac{n}{n+1} \cdot \left(1-\frac{x^2}{2} \right) + \frac{1}{n+1}\cdot \left(1-\frac{n^2x^2}{2} \right)\\
	&> \frac{n}{n+1}-\frac{1}{n+1}\cdot \frac{(n^3-n+3n)x^2}{6}+\frac{1}{n+1}-\frac{1}{n+1}\cdot \frac{n^2x^2}{2}\\
	&=1-\frac{1}{n+1}\cdot \frac{(n^3+3n^2+2n)x^2}{6}\\
	&=1-\frac{1}{n+1}\frac{(n+1)(n^2+2n)x^2}{6}\\
	&=1-\frac{(n^2+2n)x^2}{6}\\
	&=1-\frac{((n+1)^2-1)x^2}{6}
 	\end{align*}
 	The lemma now follows due to induction.
\end{proof}

Now, we prove Lemma \ref{lem:convergenceofinfiniteproduct}. 
\begin{proof}[Proof of Lemma \ref{lem:convergenceofinfiniteproduct}]
	From Lemma \ref{lem:sinnxlowerbound}, we get that,
\begin{align*}
 \left\lvert \frac{\sin(nx)}{n\sin(x)}-1 \right\rvert \leq \frac{(n^2-1)x^2}{6}.
\end{align*}
Now,
\begin{align*}
 \left\lvert \frac{\sin(nx)}{n\sin(x)}-1 \right\rvert \leq \frac{(n^2-1)x^2}{6}.
\end{align*}
This implies that,
\begin{align*}
 \sum\limits_{i=1}^{\infty} \left\lvert \frac{\sin(n\pi /2^{i+1})}{n\sin(\pi /2^{i+1})}-1 \right\rvert < \infty.
\end{align*}

Now, using Proposition 3.1 from \cite{SteinShakarchiComplex}, it follows that the infinite product is convergent. The convergence of $\prod\limits_{i=1}^{\infty} \left\lvert \cos( \pi /2^{i+1}) \right\rvert$ also follows from this argument since,
\begin{align*}
	\cos\left(\pi /2^{i+1}\right)= \frac{\sin\left(2\pi /2^{i+1}\right)}{2\sin\left(\pi /2^{i+1}\right)}.
\end{align*}

\end{proof}

Since $b_m$ is strictly increasing in $m$, any \emph{sufficiently delayed} sequence of natural numbers is a good sequence. Therefore, it is straightforward to verify that given any subset $S$ of $\N$ there is a good sequence $\langle u(m) \rangle$ containing exactly the elements of $S$ such that every element of $S$ occurs infinitely many times in $\langle u(m) \rangle$. This observation is important in our construction.

The first technical lemma is a generalization of the bound on exponential sums given in Lemma 7 from \cite{Schmidt62}.
\begin{lemma}
\label{lem:ambound}
Let $\langle u(m) \rangle_{m=1}^{\infty}$ be any good sequence of
bases greater than or equal to $2$. Let $\xi$ be the real number that
is obtained as the limit of $\langle \xi_m \rangle_{m=1}^{\infty}$
where each $\xi_m$ is chosen according to Criterion
\ref{item:criterion1} or \ref{item:criterion2}. Then, there exists a
constant $\delta$ depending only on $u(1)$ such that for every $m \geq
1$, $A_m(\xi) \leq \delta m^2 (\langle m+1 \rangle - \langle m
\rangle)^{2-\beta_m}$. 
\end{lemma}

\begin{proof}
$A_m(x)$ is equal to
\begin{align*}
A_m(x) = \sum\nolimits_{\substack{t=-m\\ t\neq 0}}^{m} \sum\nolimits_{\substack{h=1 \\ u(h) \not\sim u(m)}}^{m}\sum\nolimits_{j=\langle m;u(h) \rangle+1}^{\langle m+1;u(h) \rangle} \sum\nolimits_{g=\langle m;u(h) \rangle+1}^{\langle m+1;u(h) \rangle} e((u(h)^j-u(h)^g)tx)
\end{align*}
Let $B_m(x)$ denote the part of the sum for which either $\lvert j-g \rvert < m$ or $j$ or $g$ is at least $\langle m+1;u(h) \rangle  -m$, and let $C_m(x)$ denote the rest of the sum. We get the trivial estimate for $B_m(x)$,
\begin{align*}
B_m(x) &\leq 2m\sum\limits_{h=1}^{m} \sum'  \lvert e((u(h)^j-u(h)^g)tx) \rvert +2m \sum\limits_{h=1}^{m}\sum''\lvert e((u(h)^j-u(h)^g)tx) \rvert \\
&+  2m \sum\limits_{h=1}^{m}\sum\limits_{\substack{j = \langle m+1; u(h) \rangle-m}}^{\langle m+1;u(h) \rangle} \sum\limits_{g=\langle m;u(h) \rangle+1}^{\langle m+1;u(h) \rangle}  \lvert e((u(h)^j-u(h)^g)tx) \rvert \\
&+2m \sum\limits_{h=1}^{m}\sum\limits_{\substack{g = \langle m+1; u(h) \rangle-m}}^{\langle m+1;u(h) \rangle} \sum\limits_{j=\langle m;u(h) \rangle+1}^{\langle m+1;u(h) \rangle} \lvert e((u(h)^j-u(h)^g)tx) \rvert
\end{align*}
where $\sum'$ denotes the sum over all $j$ and $g$ between $\langle m;u(h) \rangle+1$ and $\langle m+1;u(h) \rangle$,  such that $\lvert j-g \rvert <m$ and $j > g$. Similarly, $\sum''$ denotes the sum over all $j$ and $g$ between $\langle m;u(h) \rangle+1$ and $\langle m+1;u(h) \rangle$, such that $\lvert j-g \rvert <m$ and $g \geq j$. Therefore,
\begin{align*}
B_m(x) &\leq 2m \sum\limits_{h=1}^{m} \left(\langle m+1;u(h) \rangle-\langle m;u(h) \rangle \right) + 2m \sum\limits_{h=1}^{m} \left(\langle m+1;u(h) \rangle-\langle m;u(h) \rangle \right) \\
&+ 2m^2  \sum\limits_{h=1}^{m} \left(\langle m+1;u(h) \rangle-\langle m;u(h) \rangle \right) + 2m^2  \sum\limits_{h=1}^{m} \left(\langle m+1;u(h) \rangle-\langle m;u(h) \rangle \right)\\
&= 4m  \sum\limits_{h=1}^{m} \left(\langle m+1;u(h) \rangle-\langle m;u(h) \rangle \right) + 4m^2  \sum\limits_{h=1}^{m} \left(\langle m+1;u(h) \rangle-\langle m;u(h) \rangle \right)\\
&< 8m^2  \sum\limits_{h=1}^{m} \left(\langle m+1;u(h) \rangle-\langle m;u(h) \rangle \right)\\
&< 8m^2  \sum\limits_{h=1}^{m} \left(\langle m+1 \rangle-\langle m \rangle +2 \right) \\
&\leq 16m^2  \sum\limits_{h=1}^{m} \left(\langle m+1 \rangle-\langle m \rangle \right) \\
&\leq 16m^3 \left(\langle m+1 \rangle-\langle m \rangle \right)\\
&\leq 16m^2 \left(\langle m+1 \rangle-\langle m \rangle \right)^{3/2}\\
&\leq 16m^2 \left(\langle m+1 \rangle-\langle m \rangle \right)^{2-\beta_m}.
\end{align*}

The second last inequality above follows because $m^2 \leq \langle m+1 \rangle-\langle m \rangle $. The last inequality follows since $\alpha(r,s)\leq 1/2$ for any $r$ and $s$. In the above argument we also used the fact that for any $m$ and $r$,
\begin{align*}
\langle m+1; r\rangle -\langle m; r \rangle &\leq \frac{\langle m+1 \rangle}{\log(r)} + 1 - \frac{\langle m \rangle}{\log(r)}+1\\
&= \frac{\langle m+1 \rangle-\langle m \rangle}{\log(r)} +2\\
&< \langle m+1 \rangle-\langle m \rangle +2\\
&\leq 2 (\langle m+1 \rangle-\langle m \rangle).
\end{align*}
The last inequality above follows because $\langle m+1 \rangle-\langle m \rangle \geq 2$ for any $m$. Since, the bound on $B_m(x)$ is true for any $x$, we get,
\begin{align}
\label{eq:bmbound}
\left\lvert B_m(\xi)-B_m(\xi_m) \right\rvert \leq \left\lvert B_m(\xi) \right\rvert  + \left\lvert B_m(\xi_m) \right\rvert \leq 32 \cdot  m^2 \left(\langle m+1 \rangle-\langle m \rangle \right)^{2-\beta_m}.
\end{align}
Now, we estimate $\left\lvert C_m(\xi)-C_m(\xi_m) \right\rvert$. Consider the inner term of $C_m(\xi)-C_m(\xi_m)$ with $t$ and $h$ fixed. For $m \geq M_1$, which depends only on $u(1)$, this term is at most equal to,
\begin{align*}
2 \sum\limits_{g=m}^{\substack{\langle m+1;u(h) \rangle\\-\langle m;u(h) \rangle\\-m}} \sum\limits_{j=1}^{\substack{\langle m+1;u(h) \rangle\\-\langle m;u(h) \rangle\\-m-g}} \left\lvert e(L_g u(h)^{-j}) -1\right\rvert &\leq 2 \sum\limits_{g=m}^{\substack{\langle m+1;u(h) \rangle\\-\langle m;u(h) \rangle\\-m}} \sum\limits_{j=1}^{\infty} u(h)^{-j}\\
&< 2 \left(\langle m+1;u(h) \rangle-\langle m;u(h) \rangle\right)\\
&\leq 2 \left(\langle m+1 \rangle-\langle m \rangle+1\right)\\
&\leq 4 \left(\langle m+1 \rangle-\langle m \rangle\right). 
\end{align*}
In the above expression $L_g= (u(h)^g-1) u(h)^{\langle m+1;u(h) \rangle-m-g}t(\xi-\xi_m)$. The first inequality above follows from $\lvert L_g \rvert \leq 1/2$. This is true because,
\begin{align*}
	\left\lvert L_g\right\rvert &\leq u(h)^{\langle m+1;u(h) \rangle-m}m\times u(m)^{-\langle m+1;u(m) \rangle+2}\\
	 &\leq e^{\langle m+1 \rangle}u(h)^{1-m}m \times e^{-\langle m+1 \rangle} u(m)^2\\
	 &= u(h)^{1-m}m \times u(m)^2 \\
	 &\leq \frac{u(1)^2m^4}{u(h)^m} \\
	 &\leq \frac{u(1)^2m^4}{2^m}
\end{align*}
In the above we used the fact that $u(m) \leq mu(1)$ since $\langle u(m) \rangle_{m=1}^{\infty}$ is a good sequence. From this, we obtain that $\lvert L_g \rvert \leq 1/2$ when $m \geq M_1$ for large enough $M_1$, where $M_1$ only depends on $u(1)$. Therefore, we conclude that,
\begin{align}
\label{eq:cmbound}
\left\lvert	C_m(\xi)-C_m(\xi_m) \right\rvert &\leq 8 \cdot  m^2 \left( \langle m+1 \rangle -\langle m \rangle \right)^{2-\beta_m}.
\end{align}
for $m \geq M_1$. From (\ref{eq:bmbound}) and (\ref{eq:cmbound}), we get that,
\begin{align}
\label{eq:amdistancebound}
	\lvert A_m(\xi) - A_m(\xi_m) \rvert \leq 40 \cdot m^2 \left( \langle m+1 \rangle -\langle m \rangle \right)^{2-\beta_m}
\end{align}
for $m \geq M_1$. In order to complete the proof, it is enough to show that, 
\begin{align*}
	A_m(\xi_m) \leq c \cdot m^2 \left( \langle m+1 \rangle -\langle m \rangle \right)^{2-\beta_m}
\end{align*}
 for $m\geq M$ where $c$ and $M$ are constants depending only on $u(1)$. We first consider the case when during step $m$, $\xi_m$ is chosen from $\sigma_m^*(\xi_{m-1})$ according to Criterion \ref{item:criterion1}. Let $\mu A^*_m$ denote the average of the function $A_m$ over all $p(u(m))^{b_m-a_m-2}$ elements of $\sigma_m^*(\xi_{m-1})$.  Consider the set  $\mathcal{G}_{p(u(m))}^{b_m-a_m-2}$ from Corollary \ref{cor:lowdiscrepancy}. If there exists $m'<m$ such that $u(m') \neq u(m)$ then from condition \ref{item:goodsequence4} in Definition \ref{def:goodsequence}, we know that $b_m-a_m+2 > b_m-a_m \geq \max\{N_{u(m)}(1/2),N_{p(u(m))}(1/2)\}$.  In this case
\begin{align*}
	\lvert \mathcal{G}_{p(u(m))}^{b_m-a_m-2} \rvert 	 &\geq \frac{1}{2} p(u(m))^{b_m-a_m-2} = \frac{1}{2} \left\lvert \Sigma_{p(u(m))}^{b_m-a_m-2}\right\rvert.
\end{align*}
If $\lvert\mu A^*_m \rvert \leq c \cdot m^2 \left( \langle m+1 \rangle -\langle m \rangle \right)^{2-\beta_m}$, then using the Markov inequality, we obtain $A_m(\xi_m) \leq 2c \cdot m^2 \left( \langle m+1 \rangle -\langle m \rangle \right)^{2-\beta_m}$, which proves the required inequality. If there does not exist any $m'<m$ such that $u(m') \neq u(m)$, then we know that $u(m)=u(m-1)=u(m-2)=\dots u(2)=u(1)$. But there exists at most constantly many $m$ (where the constant depends only on $u(1)$) such that,
\begin{align*}
\langle m+1 ; u(1) \rangle- \langle m;u(1) \rangle < \max\{N_{u(1)}(1/2),N_{p(u(1))}(1/2)\}.
\end{align*}
Since the above constant depends only on $u(1)$, the value of $A_m$ for these values of $m$ can be absorbed within the constant $\delta$ in the statement of the lemma. Therefore, we only address the case when there exists $m'<m$ such that $u(m') \neq u(m)$. It is enough to show that $\lvert\mu A^*_m \rvert \leq c \cdot m^2 \left( \langle m+1 \rangle -\langle m \rangle \right)^{2-\beta_m}$. We now show that this upper bound is true for some $c$. Using the bound on $B_m(x)$ proved above, we immediately get that,
\begin{align*}
\lvert \mu B^*_m \rvert \leq 16 \cdot m^2 \left(\langle m+1 \rangle-\langle m \rangle \right)^{2-\beta_m}.
\end{align*}
 Now, we give a similar upper bound for $\lvert \mu C^*_m \rvert$ to obtain the required bound on $\lvert\mu A^*_m \rvert$. Consider the average $\mu f^*$ of the function $f(x)=e(tx)$ over $\sigma_m^*(\xi_{m-1})$.
 \begin{align*}
 &\lvert \mu f^* \rvert \\
 &= \lvert e(t\eta_m(\xi_{m-1})) \rvert \prod\limits_{i=a_m+1}^{b_m-2} \left\lvert \frac{1+e(tu(m)^{-i})+e(t2u(m)^{-i})+\dots +e(t(p(u(m))-1)u(m)^{-i})}{p(u(m))}\right\rvert\\
 	&=  \prod\limits_{i=a_m+1}^{b_m-2} \left\lvert \frac{e(tp(u(m))u(m)^{-i})-1}{p(u(m))(e(tu(m)^{-i})-1)}\right\rvert\\
 	&= \prod\limits_{i=a_m+1}^{b_m-2} \left\lvert \frac{\sin( p(u(m)) \pi  t  u(m)^{-i})}{p(u(m))\sin(\pi t u(m)^{-i})} \right\rvert
 \end{align*}
    Using the above, we get that,
\begin{align*}
 &\left\lvert \mu C_m^* \right\rvert \\
 &\leq 2 \sum\limits_{\substack{t=-m\\t \neq 0}}^{m} \sum\limits_{\substack{h=1 \\ u(h) \not\sim u(m)}}^{m}\sum\limits_{g=m}^{M'} \sum\limits_{j=1}^{M'-g}  \prod\limits_{i=a_m+1}^{b_m-2} \left\lvert \frac{\sin(p(u(m))\pi (u(h)^g-1)u(h)^{\langle m;u(h) \rangle}t u(h)^j  u(m)^{-i})}{p(u(m)) \sin(\pi (u(h)^g-1)u(h)^{\langle m;u(h) \rangle}t u(h)^j  u(m)^{-i})} \right\rvert
\end{align*}
where $M'=\langle m+1;u(h) \rangle-\langle m;u(h) \rangle-m$. Fixing $t$, $h$ and $g$ and letting $L=(u(h)^g-1)u(h)^{\langle m;u(h) \rangle}t$, the inner sum is equal to,
\begin{align*}
 \sum\limits_{j=1}^{\langle m+1;u(h) \rangle-\langle m;u(h) \rangle-m-g} \prod\limits_{i=a_m+1}^{b_m-2} \left\lvert \frac{\sin(p(u(m))\pi L u(h)^j  u(m)^{-i})}{p(u(m)) \sin(\pi L u(h)^j  u(m)^{-i})} \right\rvert
\end{align*}
Now, observe that,
\begin{align*}
 \frac{Lu(h)^j}{u(m)^{b_m-2}} \leq \frac{u(h)^{\langle m+1;u(h) \rangle-m}}{u(m)^{b_m-2}} \leq \frac{u(m)^2}{u(h)^{m-1}} \leq \frac{u(1)m}{2^{m-1}} \leq \frac{1}{2}
\end{align*}
for $m \geq M_2$ where $M_2$ depends only on $u(1)$. In the above statement, we used the fact that $\langle u(m) \rangle_{m=1}^{\infty}$ is a good sequence of bases greater than or equal to $2$. Therefore, using condition \ref{item:goodsequence3} in the definition of a good sequence (Definition \ref{def:goodsequence}) we get that,
\begin{align*}
 \prod\limits_{i=b_m-1}^{\infty} \left\lvert \frac{\sin(p(u(m))\pi L u(h)^j  u(m)^{-i})}{p(u(m)) \sin(\pi L u(h)^j  u(m)^{-i})} \right\rvert &\geq \prod\limits_{i=b_m-1}^{\infty} \left\lvert \frac{\sin(p(u(m))\pi \frac{1}{2^{i+1}})}{p(u(m)) \sin(\pi \frac{1}{2^{i+1}})} \right\rvert\\
 & \geq \prod\limits_{i=1}^{\infty} \left\lvert \cos\left( \frac{\pi}{2^{i+1}}\right) \right\rvert \\
 &= \eta >0.
\end{align*}
Here, $\eta=\prod_{i=1}^{\infty} \left\lvert \cos\left( \frac{\pi}{2^{i+1}}\right) \right\rvert$ is an absolute constant. Therefore,
\begin{align*}
  &\sum\limits_{j=1}^{\langle m+1;u(h) \rangle-\langle m;u(h) \rangle-m-g} \prod\limits_{i=a_m+1}^{b_m-2} \left\lvert \frac{\sin(p(u(m))\pi L u(h)^j  u(m)^{-i})}{p(u(m)) \sin(\pi L u(h)^j  u(m)^{-i})} \right\rvert \\
  &\leq \eta^{-1}  \sum\limits_{j=1}^{\langle m+1;u(h) \rangle-\langle m;u(h) \rangle-m-g} \prod\limits_{i=a_m+1}^{\infty} \left\lvert \frac{\sin(p(u(m))\pi L u(h)^j  u(m)^{-i})}{p(u(m)) \sin(\pi L u(h)^j  u(m)^{-i})} \right\rvert
\end{align*}

Now,  
\begin{align*}
	\lvert L \rvert &\geq (u(h)^g-1)u(h)^{\langle m;u(h) \rangle}\\
	& \geq (2^m-1)e^{\langle m \rangle}\\
	& > (2^m-1) u(m)^{a_m-1}\\
	&=u(m)^{a_m+1}\frac{2^m-1}{u(m)^2}\\
	&\geq u(m)^{a_m+1}\frac{2^m-1}{u(1)^2 m^2}\\
	&>u(m)^{a_m+1}
\end{align*}
for $m>M_3$ where $M_3$ depends only on $u(1)$. In the above argument,  $u(m)^2 \leq m^2 u(1)^2$ follows from property \ref{item:goodsequence3} in the definition of good sequences (Definition \ref{def:goodsequence}). Since $\lvert L \rvert \geq u(m)^{a_m+1}$, we apply Lemma \ref{lem:sinbound} to obtain,

\begin{align*}
 &\eta^{-1} \sum\limits_{j=1}^{\langle m+1;u(h) \rangle-\langle m;u(h) \rangle-m-g} \prod\limits_{i=a_m+1}^{\infty} \left\lvert \frac{\sin(p(u(m))\pi L u(h)^j  u(m)^{-i})}{p(u(m)) \sin(\pi L u(h)^j  u(m)^{-i})} \right\rvert \\
 &\leq 2\eta^{-1} \left(\langle m+1;u(h) \rangle-\langle m;u(h) \rangle-m-g \right)^{1-\alpha(u(h),u(m))}\\
 &\leq 2\eta^{-1} \left(\langle m+1;u(h) \rangle-\langle m;u(h) \rangle \right)^{1-\alpha(u(h),u(m))}\\
 &\leq  2\eta^{-1} \left(\langle m+1;u(h) \rangle-\langle m;u(h) \rangle \right)^{1-\beta_m}\\
 &< 2\eta^{-1} \left(\langle m+1 \rangle-\langle m \rangle+1 \right)^{1-\beta_m}\\
 &\leq 2\eta^{-1}\cdot 2^{1-\beta_m} \left(\langle m+1 \rangle-\langle m \rangle \right)^{1-\beta_m}\\
 &\leq 2\eta^{-1}\cdot 2 \left(\langle m+1 \rangle-\langle m \rangle \right)^{1-\beta_m}\\
 &= 4\eta^{-1}\cdot \left(\langle m+1 \rangle-\langle m \rangle \right)^{1-\beta_m}\\
\end{align*}
Therefore,
\begin{align*}
 \lvert \mu C_m^* \rvert &\leq 2m \cdot m \cdot 2 (\langle m+1 \rangle-\langle m \rangle)\cdot 4\eta^{-1}\cdot \left(\langle m+1 \rangle-\langle m \rangle \right)^{1-\beta_m}\\
 &\leq 16  \eta^{-1} \cdot \left(\langle m+1 \rangle-\langle m \rangle \right)^{2-\beta_m} 
\end{align*}
for $m \geq \max\{M_2,M_3\}$. Therefore, we obtain that for for $m \geq \max\{M_2,M_3\}$,
\begin{align}
\label{eq:amaveragebound}
\lvert\mu A^*_m \rvert \leq 32 \eta^{-1} \cdot m^2 \left( \langle m+1 \rangle -\langle m \rangle \right)^{2-\beta_m}.
\end{align}
Since $M_1$, $M_2$ and $M_3$ are constants that only depend on $u(1)$, using (\ref{eq:amdistancebound}) and (\ref{eq:amaveragebound}), we get that there exists a constant $\delta_1$ depending only on $u(1)$ such that,
\begin{align*}
	A_m(\xi) \leq \delta_1 m^2 (\langle m+1 \rangle - \langle m \rangle)^{2-\beta_m}
\end{align*} 
for any $m \geq 1$ when $\xi_m$ is chosen from $\sigma^*_m(\xi_{m-1})$ according to Criterion \ref{item:criterion1}. The case when $\xi_m$ is chosen from $\sigma_m(\xi_{m-1})$ according to Criterion \ref{item:criterion2} is handled in a very similar way since this is equivalent to the special case when $p(u(m))=u(m)$ itself. In this case we obtain a constant $\delta_2$ depending only on $u(1)$ such that,
\begin{align*}
	A_m(\xi) \leq \delta_2 m^2 (\langle m+1 \rangle - \langle m \rangle)^{2-\beta_m}
\end{align*} 
for any $m \geq 1$ when $\xi_m$ is chosen from $\sigma_m(\xi_{m-1})$ according to Criterion \ref{item:criterion2}. Letting $\delta=\max\{\delta_1,\delta_2\}$, we obtain the final bound,
\begin{align*}
 A_m(\xi_m) \leq \delta m^2 \left( \langle m+1 \rangle -\langle m \rangle \right)^{2-\beta_m}
\end{align*}
for every $m \geq 1$, which completes the proof of the lemma since $\delta$ depends only on $u(1)$.
\end{proof}
 
We make crucial use of the following lemma regarding the \emph{uniform normality} of the infinite sequences constructed by choosing successive set of digits from $\mathcal{G}_j^{b_m-a_m-2}$, in the proof of Theorem \ref{thm:maintheorem}. 

\begin{lemma}
\label{lem:lowdiscrepancycombinationnormality}
Let $b$ be any base and $j \leq b$. For any finite string $w \in
\Sigma_j^*$, $\epsilon>0$, there exists an integer
$L'_{b,j}(w,\epsilon)$ satisfying the following property. If $X$ is
any infinite sequence in $\Sigma_j^\infty$ such that,$X_{\langle m;b
  \rangle +1}X_{\langle m;b \rangle +2}\dots X_{\langle m+1 ;b \rangle
  -2} \in \mathcal{G}_j^{\langle m+1 ;b \rangle -\langle m;b \rangle
  -2} $ for every $m> 0$ and if $T\geq 0$ is any non-negative integer,
then for all $n \geq \langle T;b \rangle+L'_{b,j}(w,\epsilon)$,
$\lvert P(X_{\langle T;b \rangle+1}^n,w)- j^{-\lvert w \rvert} \rvert
\leq \epsilon$.

\end{lemma}
Note especially that $L'_{b,j}(w,\epsilon)$ is a constant which
depends only on $w$ and $\epsilon$. This constant is independent of
the infinite sequence $X$ and the \emph{starting block number}
$T$. When $T=0$, we have $\langle 0;b \rangle=0$ and hence the above
statement asserts that, $\lvert N(w,X_{1}^n)/n-j^{-\lvert w \rvert}
\rvert \leq \epsilon$ for all $n \geq L'_{b,j}(w,\epsilon)$.

\begin{proof}[Proof of Lemma \ref{lem:lowdiscrepancycombinationnormality}]
The conclusion of the lemma is equivalent to the existence of a number
$L'_{b,j}(w,\epsilon)$ such that
\begin{align*}
\left\lvert \frac{N(w,X_{\langle T;b \rangle+1}^n)}{n-\langle T;b
  \rangle- |w|+1} - \frac{1}{j^{\lvert w \rvert}}
\right\rvert \leq \epsilon	 
\end{align*}
for all $n \geq \langle T;b \rangle+L'_{b,j}(w,\epsilon)$. We prove
this equivalent statement below.

To abbreviate the expressions in the proof, let $\hat\ell_w = |w|-1$.

\newcommand\he{\hat\ell_w}.

The number of $\lvert w \rvert$-length blocks up to some $\langle
M+1;b \rangle$ containing the indices $\langle m+1 ;b \rangle -2$,
$\langle m+1 ;b \rangle -1$ or $\langle m+1 ;b \rangle$ for some $m
\leq M$ is at most $(\lvert w \rvert+2)M$. Therefore, the fraction of
such digits in the first $\langle M+1;b \rangle$ digits is at most 
\begin{align*}
\frac{(\lvert w \rvert +2)M}{\langle M+1;b \rangle} \leq \frac{(\lvert
  w \rvert +2)M\cdot \log(b)}{e^{\sqrt{M+1}}}. 
\end{align*}
Let $M_1(w,\epsilon)$ be large enough so that the term on the right
hand side above is at most $\epsilon$. Let $C_j$ be the constant
corresponding to base $j$ from Corollary \ref{cor:lowdiscrepancy}, and
let $N^j_1(w,\epsilon)$ be any large enough number such that for all
$n \geq N^j_1(w,\epsilon)$, 
\begin{align*}
C_j \frac{\sqrt{\log \log n}}{\sqrt{n}} \leq \epsilon.	
\end{align*}
For any $m>0$, if $\langle m+1 ;b \rangle -\langle m;b \rangle -2 \geq
N_b(1/2)+N^j_1(w,\epsilon)+\he$, then from Corollary
\ref{cor:lowdiscrepancy} it follows that 
\begin{align*}
\left\lvert \frac{N(w,X_{\langle m;b \rangle +1}^{\langle m+1 ;b \rangle -2})}{\langle m+1 ;b \rangle -\langle m;b \rangle -\he} - \frac{1}{j^{\lvert w \rvert}} \right\rvert	\leq \epsilon.
\end{align*}

Since $\langle m+1 ;b \rangle -\langle m;b \rangle $ is increasing in $m$, $\langle m+1 ;b \rangle -\langle m;b \rangle -2$ is greater than $N_b(1/2)+N^j_1(w,\epsilon)+\he$ for all but finitely many $m$. Let $M_2(w,\epsilon)$ denote the smallest integer such that $\langle m+1 ;b \rangle -\langle m;b \rangle -2 \geq N_b(1/2)+N^j_1(w,\epsilon)+\he$ for every $m \geq M_2 (w,\epsilon)$. 

Consider any $n \geq \langle M_2(w,\epsilon);b \rangle$. Let $M'$ be the (unique) integer such that $\langle M'; b \rangle \leq n < \langle M'+1;b \rangle$. We have $M' \geq M_2(w,\epsilon)$. We first consider the case when $n=\langle M'; b \rangle$. Then,
\begin{align*}
&\left\lvert \frac{N(w,X_{\langle T;b \rangle+1}^n)}{n-\langle T;b \rangle-\he+2} - \frac{1}{j^{\lvert w \rvert}} \right\rvert \\
&\leq \frac{S(w,\epsilon)}{n-\langle T;b \rangle-\he+2} \sum\limits_{m=T}^{ M_2(w,\epsilon)-1} \frac{\langle m+1 ;b \rangle -\langle m;b \rangle -\he}{S(w,\epsilon)}	\left\lvert \frac{N(w,X_{\langle m;b \rangle +1}^{\langle m+1 ;b \rangle -2})}{\langle m+1 ;b \rangle -\langle m;b \rangle -\he} - \frac{1}{j^{\lvert w \rvert}} \right\rvert \\
&+ \frac{S'(w,\epsilon)}{n-\langle T;b \rangle-\he+2} \sum\limits_{m=M_2(w,\epsilon)}^{ M'-1} \frac{\langle m+1 ;b \rangle -\langle m;b \rangle -\he}{S'(w,\epsilon)}	\left\lvert \frac{N(w,X_{\langle m;b \rangle +1}^{\langle m+1 ;b \rangle -2})}{\langle m+1 ;b \rangle -\langle m;b \rangle -\he} - \frac{1}{j^{\lvert w \rvert}} \right\rvert \\
&+ \left(1- \frac{S(w,\epsilon)+S'(w,\epsilon)}{n-\langle T;b \rangle-\he+2} \right) \left( 1+\frac{1}{j^{\lvert w \rvert}} \right) 
\end{align*}
where, 
\begin{align*}
S(w,\epsilon) = \sum\limits_{T \leq m < M_2(w,\epsilon)} \langle m+1 ;b \rangle -\langle m;b \rangle -\he	
\end{align*}
and,
\begin{align*}
S'(w,\epsilon) = \sum\limits_{ M_2(w,\epsilon) \leq m < M'} \langle m+1 ;b \rangle -\langle m;b \rangle -\he.
\end{align*}

 From the definition of $M_2(w,\epsilon)$ and the properties of $N_b(1/2)$ and $N^j_1(w,\epsilon)$, we have,
\begin{align*}
	\left\lvert \frac{N(w,X_{\langle m;b \rangle +1}^{\langle m+1 ;b \rangle -2})}{\langle m+1 ;b \rangle -\langle m;b \rangle -\he} - \frac{1}{j^{\lvert w \rvert}} \right\rvert \leq \epsilon
\end{align*}
for every $m \geq M_2(w,\epsilon)$. 

Let $E(w,M')$ is the proportion of $\lvert w \rvert$-length blocks among all the $\lvert w \rvert$-length blocks in $X_{\langle T;b \rangle+1}^{n}$ containing the indices $\langle m+1 ;b \rangle -2$, $\langle m+1 ;b \rangle -1$ or $\langle m+1 ;b \rangle $ for some $m < M'$. If $M' \geq M_1(w,\epsilon)$ then $E(w,M')\leq \epsilon$. Therefore, for $n \geq \langle M_1(w,\epsilon); b \rangle$
\begin{align*}
\left(1- \frac{S(w,\epsilon)+S'(w,\epsilon)}{n-\langle T;b \rangle-\he+2} \right) \leq E(w,M') \leq \epsilon.	
\end{align*}

Observe that,
\begin{align*}
\frac{S(w,\epsilon)}{n-\langle T;b \rangle-\he+2} < \frac{\sum\limits_{0 \leq m < M_2(w,\epsilon)} \left( \langle m+1 ;b \rangle -\langle m;b \rangle -\he\right)	}{n-\langle T;b \rangle-\he+2}.
\end{align*}
Therefore, there exists a large enough number $M_3(w,\epsilon)$ such that if $n-\langle T;b \rangle \geq M_3(w,\epsilon)$, then,
\begin{align*}
\frac{S(w,\epsilon)}{n-\langle T;b \rangle-\he+2} < \epsilon.
\end{align*}

Let $N_T$ denote the following quantity.
\begin{align*}
N_T 
=
\left\lvert \frac{N(w,X_{\langle T;b \rangle+1}^n)}{n-\langle T;b
  \rangle-\he+2}-\frac{1}{j^{\lvert w \rvert}} \right\rvert.
\end{align*}

From the above observations, it follows that if $n \geq \max\{\langle M_1(w,\epsilon);b \rangle,\langle M_2(w,\epsilon);b \rangle,M_3(w,\epsilon)\}$, then,
\begin{align*}
N_T
&\leq \epsilon \\
&+ \frac{S'(w,\epsilon)}{n-\langle T;b \rangle-\he+2} \sum\limits_{m=M_2(w,\epsilon)}^{M'} \frac{\langle m+1 ;b \rangle -\langle m;b \rangle -\he}{S'(w,\epsilon)}	\left\lvert \frac{N(w,X_{\langle m;b \rangle +1}^{\langle m+1 ;b \rangle -2})}{\langle m+1 ;b \rangle -\langle m;b \rangle -\he} - \frac{1}{j^{\lvert w \rvert}} \right\rvert \\
&+ 2\epsilon  \\
&\leq \frac{S'(w,\epsilon)}{n-\langle T;b \rangle-\he+2} \epsilon + 3 \epsilon \\
&\leq 4\epsilon.
\end{align*}
Therefore in the case when $n \geq \max\{\langle M_1(w,\epsilon);b \rangle,\langle M_2(w,\epsilon);b \rangle,M_3(w,\epsilon)\}$ and $n=\langle M'; b \rangle$, we have,
\begin{align*}
N_T \leq 4\epsilon.
\end{align*}
Now, when $n > \langle M'; b \rangle$, then we also have to consider the $\lvert w \rvert$-length blocks in $X_{\langle M'; b \rangle-\lvert w \rvert+2}^n$. We have,
\begin{align*}
N_T 
&= \frac{\langle M'; b \rangle-\langle T;b \rangle-\he+2}{n-\langle T;b \rangle-\he+2} \frac{N(w,X_{\langle T;b \rangle+1}^{\langle M'; b \rangle})}{\langle M'; b \rangle-\langle T;b \rangle-\he+2}
 &+ \frac{n-\langle M'; b \rangle}{n-\langle T;b \rangle-\he+2} \frac{N(w,X_{\langle M'; b \rangle-\lvert w \rvert +2}^{n})}{n-\langle M'; b \rangle}. 
\end{align*}
Let $M_4(w,\epsilon)$ be the smallest integer such that,
\begin{align*}
\langle M_4(w,\epsilon); b \rangle \geq \left\lceil \frac{N_b(1/2)+N^j_1(w,\epsilon)+\he}{\epsilon} \right\rceil.
\end{align*}
Now, let
\begin{align*}
L'_{b,j}(w,\epsilon) = \max\left\{\langle M_1(w,\epsilon);b \rangle,\langle M_2(w,\epsilon);b \rangle,M_3(w,\epsilon),\langle M_4(w,\epsilon); b \rangle,(\he+2)/\epsilon+\lvert w \rvert\right\}	.
\end{align*}
Consider any $n \geq \langle T;b \rangle+ L'_{b,j}(w,\epsilon)$. We showed that
\begin{align*}
\left\lvert \frac{N(w,X_{\langle T;b \rangle+1}^{\langle M'; b \rangle})}{\langle M'; b \rangle-\langle T;b \rangle-\he+2} - \frac{1}{j^{\lvert w \rvert}} \right\rvert \leq 4\epsilon.
\end{align*}

Then, for $n$ ranging from $\langle M'; b \rangle+1$ to $\langle M'; b \rangle+\lvert w \rvert -2$, 
\begin{align*}
N_T \quad
&\leq \frac{\langle M'; b \rangle-\langle T;b \rangle-\he+2}{n-\langle T;b \rangle-\he+2} \left\lvert \frac{N(w,X_{\langle T;b \rangle+1}^{\langle M'; b \rangle})}{\langle M'; b \rangle-\langle T;b \rangle-\he+2} -\frac{1}{j^{\lvert w \rvert}} \right\rvert \\
&\quad+ \frac{n-\langle M'; b \rangle}{n-\langle T;b \rangle-\he+2}
\left(1+\frac{1}{j^{\lvert w \rvert}} \right)\\ 
&\leq 4\epsilon + \frac{n-\langle M'; b \rangle}{n-\langle T;b
  \rangle-\he+2} \left(1+\frac{1}{j^{\lvert w \rvert}} \right) \\ 
&\leq 4\epsilon + \frac{\he}{n-\langle T;b
  \rangle-\he+2}\left(1+\frac{1}{j^{\lvert w \rvert}} \right)\\ 
&\leq 6 \epsilon.
\end{align*}
The last inequality above follows because $n-\langle T;b \rangle \geq L'_{b,j}(w,\epsilon) \geq \lvert w \rvert/\epsilon +\he$.

For $n \geq \langle M'; b \rangle+\he$, we first consider the case when $n-\langle M'; b \rangle < N_b(1/2)+ N^j_{1}(w,\epsilon)+ \he$. The number of $\lvert w \rvert$-length blocks in the portion of $X$ between $\langle M'; b \rangle$ and $n$ is at most an $\epsilon$-fraction of the total number of $\lvert w \rvert$-length blocks up to $n$. This follows from the definition of $M_4(w,\epsilon)$ and the fact that $n-\langle T;b \rangle \geq \langle M_4(w,\epsilon); b \rangle$. Therefore, we get that
\begin{align*}
N_T  \quad
&\leq 
\frac{\langle M'; b \rangle-\langle T;b \rangle-\he+2}
     {n-\langle T;b \rangle-\he+2} 
\left\lvert 
\frac{N(w,X_{\langle T;b \rangle+1}^{\langle M'; b \rangle})}
     {\langle M'; b \rangle-\langle T;b \rangle-\he+2} -
\frac{1}{j^{\lvert w \rvert}} \right\rvert \\
&\quad+ \frac{\he}{n-\langle T;b \rangle-\he+2} 
\left(1+\frac{1}{j^{\lvert w \rvert}} \right)+ 2\epsilon \\
&\leq 4 \epsilon + 2\epsilon + 2\epsilon \\
&= 8\epsilon. 
\end{align*}
Finally we consider the case when $n-\langle M'; b \rangle \geq N_b(1/2)+ N^j_{1}(w,\epsilon)+ \he$. In this case, the probability of occurrence of $w$ in $X_{\langle M'; b \rangle+1}^{n}$ is $\epsilon$-close to $\frac{1}{j^{\lvert w \rvert}}$. This easily follows from the definition of $N^j_1(w,\epsilon)$. Therefore, we get estimates similar to those we used in the case when $n-\langle M'; b \rangle < N_b(1/2)+ N^j_{1}(w,\epsilon)+ \he$ yielding the following bound.
\begin{align*}
N_T \quad &= \frac{\langle M'; b \rangle-\langle T;b
  \rangle-\he+2}{n-\langle T;b \rangle-\he+2}
\left\lvert\frac{N(w,X_{\langle T;b \rangle+1}^{\langle M'; b
    \rangle})}{\langle M'; b \rangle-\langle T;b
  \rangle-\he+2}-\frac{1}{j^{\lvert w \rvert}} \right\rvert \\ &\quad+
\frac{n-\langle M'; b \rangle}{n-\langle T;b \rangle-\he+2}
\left\lvert\frac{N(w,X_{\langle M'; b \rangle-\lvert w \rvert
    +2}^{n})}{n-\langle M'; b \rangle}-\frac{1}{j^{\lvert w \rvert}}
\right\rvert \\ &\leq 4 \epsilon + \frac{\he}{n-\langle T;b
  \rangle-\he+2} \left(1+\frac{1}{j^{\lvert w \rvert}} \right)
\\ &\quad+ \frac{n-\langle M'; b \rangle-\he+2}{n-\he+2}
\left\lvert\frac{N(w,X_{\langle M'; b \rangle+1}^{n})}{n-\langle M'; b
  \rangle-\he+2}-\frac{1}{j^{\lvert w \rvert}} \right\rvert \\ &\leq
6\epsilon + \frac{n-\langle M'; b \rangle-\he+2}{n-\he+2} \epsilon
\\ &\leq 7\epsilon.
\end{align*}
Hence, we showed that $L'_{b,j}(w,\epsilon)$ is an integer which satisfies the properties claimed in the statement of the lemma. The proof of the lemma is thus complete.
\end{proof}

The following is an immediate corollary of the above lemma.

\begin{corollary}
\label{cor:lowdiscrepancycombinationnormality}
Let $b$ be any base and $j \leq b$. For any $k>0$ and $\epsilon>0$,
there exists an integer $L_{b,j}(k,\epsilon)$ satisfying the following
property. If $X$ is any infinite sequence in $\Sigma_j^\infty$ such
that,$X_{\langle m;b \rangle +1}X_{\langle m;b \rangle +2}\dots
X_{\langle m+1 ;b \rangle -2} \in \mathcal{G}_j^{\langle m+1 ;b
  \rangle -\langle m;b \rangle -2} $ for every $m> 0$ and if $T\geq 0$
is any non-negative integer, then for every $w \in \Sigma_j^{*}$ with
$\lvert w \rvert \leq k$ and all $n \geq \langle T;b
\rangle+L_{b,j}(k,\epsilon)$, the inequality in Lemma
\ref{lem:lowdiscrepancycombinationnormality} holds.
\end{corollary}

\section{Main construction}
\label{sec:mainconstruction}

The construction for proving Theorem \ref{thm:maintheorem} consists of
multiple stages. In the $k$\textsuperscript{th} stage, we \emph{fix}
digits in base $v(k)$ by fixing elements of the sequence $\langle u(m)
\rangle$ to be $v(k)$. Each stage consists of two substages, which have
several steps. In both the substages, we fix $u(m)=v(k)$ for
sufficiently large number of steps $m$. During the first substage, we choose $\xi_m$ from
$\sigma_{m}^*(\xi_{m-1})$ according to Criterion
\ref{item:criterion1} and during the second substage, we choose $\xi_m$ from
$\sigma_{m}(\xi_{m-1})$ according to Criterion \ref{item:criterion2}.

We ensure that $\langle u(m) \rangle$ is a good sequence of natural
numbers. During stage $1$, $u(1)$ is set to $2^{d_2}$, since
$r_1=2$. All the conditions in the definition of a good sequence are
trivially satisfied at this stage. Further checks are performed at the
end of the second substage of every stage $k$ to ensure that on
transitioning to stage $k+1$, where $u(m)$ is set to $v(k+1)$, none of
the conditions in Definition \ref{def:goodsequence} are being
violated.

The lengths of the substages are controlled carefully so that all the
requirements given in section \ref{sec:overview} are satisfied. We
describe the construction of the two substages of the
$k$\textsuperscript{th} stage in the subsections below.

\subsection{First substage of the $k$\textsuperscript{th} stage}
\label{subsec:firstsubstage}
For any $\epsilon > 0$, there exists a constant $\delta_k(\epsilon)$\label{text:deltakepsilondefinition} satisfying the following: If $\mu_1$ and $\mu_2$ are probability distributions over $\Sigma_{v(k)}^l$ for any $l \leq k$ such that
$\left\lvert \mu_1(w)-\mu_2(w) \right\rvert \leq \delta_{k}(\epsilon)$ for every $w \in \Sigma_{v(k)}^l$, then $\left\lvert H(\mu_1) - H(\mu_2) \right\rvert \leq \epsilon$. The existence of $\delta_k(\epsilon)$ follows from the uniform continuity of the Shannon entropy function (see \cite{Khinchin57}, \cite{CovTho91}).

In the first substage of stage $k$ we set $u(m)=v(k)$ for sufficiently
large number of $m$'s and choose $\xi_m$ from
$\sigma_{m}^*(\xi_{m-1})$ according to Criterion
\ref{item:criterion1}. As in section \ref{sec:overview}, let $P_k^1$
denote the index of the last step in the first substage of stage
$k$. Recall that $F_k^1=\langle P_k^1+1;v(k) \rangle$ denotes the
index of the final digit in $X(k)$ fixed during the first substage of
stage $k$. We make $P_k^1$ large enough so that the following
conditions are satisfied:
\begin{enumerate}
\label{text:firstsubstageconditions}	\item\label{item:firstsubstagecond1} $\lvert H_l^{v(k)}(X(k)_1^n) -q_{r_k} \rvert \leq 2^{-k}$ for every $l \leq k$ when $n=F_k^1$.
	\item\label{item:firstsubstagecond2} For every $m \geq P_k^1$,
          \begin{align}
\label{eq:firsttosecondsubstagetransition}
b_m-a_m \geq \frac{ L_{v(k),v(k)}(k,\delta_k(2^{-k})/2)+2k}{\min\{\delta_k(2^{-k}),2^{-k}\}/2}+k.	
\end{align}
\end{enumerate}
The constant $L_{v(k),v(k)}(k,\delta_k(2^{-k})/2)$ in condition \ref{item:firstsubstagecond2} is from Corollary \ref{cor:lowdiscrepancycombinationnormality}. Recall that $v^*(k) =
 r_k^{e_{r_k}}$. Condition \ref{item:firstsubstagecond1} is satisfied for large enough
 $P_k^1$ because the occurrence probability of any finite string $w$
 in alphabet $\{0,1,\dots ,v^*(k)-1\}$ converges to $v^*(k)^{-\lvert w
   \rvert}$ on choosing $\xi_m$ from $\sigma_{m}^*(\xi_{m-1})$
 according to Criterion \ref{item:criterion1} for sufficiently large
 number of $m$'s. This follows as a consequence of Lemma
 \ref{lem:lowdiscrepancycombinationnormality}.
  Since $b_m-a_m \geq (e^{\sqrt{m+1}}-e^{\sqrt{m}}+m^2)/\log(v(k))-1$
  and the right hand side of (\ref{eq:firsttosecondsubstagetransition})
  is a constant depending only on $k$, condition
  \ref{item:firstsubstagecond2} is satisfied for all sufficiently
  large $m$.

\subsection{Second substage of the $k$\textsuperscript{th} stage}
\label{subsec:secondsubstage}
In the second substage, we set $u(m)=v(k)$ in every step $m$ and choose $\xi_m$ from $\sigma_m(\xi_{m-1})$ according to Criterion \ref{item:criterion2}. In order to describe the construction of the second substage, we need the following technical lemmas.
\begin{lemma}
\label{lem:exponentialtermsstrongbound} Let $b$ be an arbitrary base and $\epsilon > 0$. Let $\langle c_i \rangle_{i=1}^{\infty}$ be any non-increasing sequence of real numbers in $[0,1]$ such that $c_1=1/4$ and $c_i \geq c_1/\sqrt[4]{i}$. Let $\delta$ be the constant from Lemma \ref{lem:ambound}. Then, there exists a large enough number $M(\epsilon,b)$ depending only on $\epsilon$ and $b$ satisfying the following. For $m\geq M(\epsilon,b)$ and any $l \leq \langle m+1; b \rangle - \langle m;b \rangle$,
\begin{align*}
(\langle m;b \rangle+l)^{-1}\left( \delta m \sum\nolimits_{i=1}^{m-1} (\langle i+1 \rangle-\langle i \rangle)^{1-c_i}+l\right) \leq \epsilon.	
\end{align*}
\end{lemma}
\begin{proof}
	We first consider the case when $l=0$. In this case, we have,
	\begin{align*}
\delta m \sum\limits_{i=1}^{m-1} (\langle i+1 \rangle-\langle i \rangle)^{1-c_i}
&\leq \delta m \sum\limits_{i=1}^{m-1} \left( \langle i+1 \rangle -\langle i \rangle \right)^{1-c_m}	\\
&\leq \delta m^2  \langle m \rangle ^{1-c_m}.
\end{align*}
The second inequality follows due to the H\"older's inequality with $p=1/(1-c_m)$ and $q=1/c_m$. Therefore,
\begin{align*}
\frac{1}{\langle m;b \rangle}\left( \delta m \sum\limits_{i=1}^{m-1} (\langle i+1 \rangle-\langle i \rangle)^{1-c_i}\right) &\leq \delta \log b\frac{m^2}{\langle m \rangle^{c_m}}\\
&\leq \delta \log b \frac{m^2}{\langle m \rangle^{c_m}}\\
&\leq \delta \log b  \frac{m^2}{e^{\sqrt{m}c_m}}.
\end{align*}
We have $c_m \geq \frac{c_1}{\sqrt[4]{m}}=\frac{1}{4\sqrt[4]{m}}$. Since $m^2=o(e^{\sqrt[4]{m}/4})$, the right hand side term above converges to $0$ for large enough $m$. The speed of convergence of this term to $0$ is independent of the sequence $\langle c_i \rangle$, and hence given any $\epsilon > 0$, there exists a large enough number $M'(\epsilon,b)$ such that for every $m \geq M'(\epsilon,b)$,

\begin{align}
\label{eq:weylaveragebound1}
\frac{1}{\langle m;b \rangle}\left( \delta m \sum\limits_{i=1}^{m-1} (\langle i+1 \rangle-\langle i \rangle)^{1-c_i}\right) &\leq 	\epsilon.
\end{align}
Now we consider the case when $l \neq 0$. Consider any $l \leq \langle m+1 ; b \rangle - \langle m ; b \rangle$,

\begin{align*}
\frac{\langle m+1 ; b \rangle - \langle m ; b \rangle}{\langle m ; b \rangle+l} &\leq \frac{\langle m+1 ; b \rangle - \langle m ; b \rangle}{\langle m ; b \rangle}\\
&\leq \frac{\frac{\langle m+1 \rangle}{\log b} - \frac{\langle m+1 \rangle}{\log b}+1}{\frac{\langle m \rangle}{\log b}}\\
&= \frac{\log(b)+\langle m+1 \rangle - \langle m \rangle}{\langle m \rangle}.
\end{align*}

If $m >b$, we get that,

\begin{align*}
\frac{\langle m+1 ; b \rangle - \langle m ; b \rangle}{\langle m ; b \rangle+l} 
&\leq \frac{\log(m)+ 1+ e^{\sqrt{m+1}}+2u(1)(m+1)^3 - e^{\sqrt{m}}-2u(1)m^3}{e^{\sqrt{m}}+2u(1)m^3}\\
&= \frac{\log(m)+ 1}{e^{\sqrt{m}}+2u(1)m^3} + \frac{ e^{\sqrt{m+1}}- e^{\sqrt{m}}}{e^{\sqrt{m}}+2u(1)m^3} + \frac{2u(1)(m+1)^3 - 2u(1)m^3}{e^{\sqrt{m}}+2u(1)m^3}\\
&= \frac{\log(m)+ 1}{e^{\sqrt{m}}+2u(1)m^3} + \frac{e^{\sqrt{m}}}{e^{\sqrt{m}}+2u(1)m^3} \cdot  (e^{\sqrt{m+1}-\sqrt{m}}- 1) \\
&+ \frac{2u(1)(m+1)^3 - 2u(1)m^3}{e^{\sqrt{m}}+2u(1)m^3}\\
&\leq \frac{\log(m)+ 1}{e^{\sqrt{m}}+2u(1)m^3} +   (e^{\sqrt{m+1}-\sqrt{m}}- 1) + \frac{2u(1)(m+1)^3 - 2u(1)m^3}{e^{\sqrt{m}}+2u(1)m^3}\\
&\leq \frac{\log(m)+ 1}{e^{\sqrt{m}}} +   (e^{\sqrt{m+1}-\sqrt{m}}- 1) + \frac{2u(1)((m+1)^3 - m^3)}{2u(1)m^3}\\
&\leq \frac{2\log(m)}{e^{\sqrt{m}}} +   (e^{\sqrt{m+1}-\sqrt{m}}- 1) + \frac{3m^2+3m+1}{m^3}. 
\end{align*}

Since,
\begin{align*}
 e^{\sqrt{m+1}-\sqrt{m}}- 1 &= e^{\frac{1}{\sqrt{m+1}+\sqrt{m}}}- 1
\end{align*}
this term is less than any given $\epsilon$ for $m \geq M'$ where $M'$ is the least number such that,
\begin{align*}
 \frac{1}{\sqrt{M'+1}+\sqrt{M'}} \leq \frac{\epsilon}{2}.
\end{align*}
This is easily verified by using the expansion for $e^x$. The first and last terms above also goes to $0$ as $m \to 0$ with a speed of convergence that is independent of the sequence $\langle c_i \rangle$. 

From these observations, we conclude that for any $\epsilon>0$, there exists a large enough number $M''(\epsilon,b)$ such that for any $m \geq M''(\epsilon,b)$ and $l \leq \langle m+1 ; b \rangle - \langle m ; b \rangle$, 
\begin{align}
\label{eq:weylaveragebound2}
\frac{l}{\langle m ; b \rangle+l} \leq \frac{\langle m+1 ; b \rangle - \langle m ; b \rangle}{\langle m ; b \rangle+l} < \epsilon.
\end{align}

Let $M(\epsilon,b)=\max\{M'(\epsilon/2,b),M''(\epsilon/2,b)\}$. From (\ref{eq:weylaveragebound1}) and (\ref{eq:weylaveragebound2}), we conclude that for $m \geq M(\epsilon,b)$ and any $l \leq \langle m+1 ; b \rangle - \langle m ; b \rangle$,
\begin{align*}
\frac{1}{\langle m;b \rangle+l}&\left( \delta m \sum\limits_{i=1}^{m-1} (\langle i+1 \rangle-\langle i \rangle)^{1-c_i}+l\right)\\
 &\leq \frac{1}{\langle m;b \rangle+l}\left( \delta m \sum\limits_{i=1}^{m-1} (\langle i+1 \rangle-\langle i \rangle)^{1-c_i}+\langle m+1 ; b \rangle - \langle m ; b \rangle\right)\\
 &\leq \frac{\epsilon}{2}+\frac{\epsilon}{2} \\
 &= \epsilon.
\end{align*}

\end{proof}

\begin{lemma}
\label{lem:smallweylaverageimpliesentropy1}
	Let $b$ be an arbitrary base, $k$ be any natural number and $\epsilon > 0$.  There exists a large enough integer $T(\epsilon,b,k)$ and a positive real number $\gamma(\epsilon,b,k)$ satisfying the following. Let $x$ is any real number in $[0,1]$ having base-$b$ expansion $X \in \Sigma_b^\infty$. If $ \lvert \frac{1}{n} \sum_{i=1}^{n} e(tb^{(j-1)}x) \rvert < 	\gamma(\epsilon,b,k)$ for every $t$ with $\lvert t \rvert \leq T(\epsilon,b,k)$, then, $\left\lvert H_l^{b}(X_1^n)-1\right\rvert < \epsilon$
for every $l \leq k$.
\end{lemma}

\begin{proof}
	In order to prove the lemma, we use the well-known Fourier expansion of characteristic functions of cylinder sets given in Lemma \ref{lem:koksmaapproximation} (\cite{Koksma74}). We show that for every $w \in \Sigma_b^*$ and $\epsilon'>0$,
		\begin{align*}
 \left\lvert \frac{N(w,Y_1^{n})}{n-\lvert w \rvert +1} - \frac{1}{b^{\lvert w \rvert}} \right\rvert < \epsilon'
\end{align*}
if the Weyl averages are \emph{sufficiently small} for \emph{sufficiently many} values of $t$, where the parameters involved only depends on $b$, $\epsilon'$ and $w$. The lemma follows easily from this claim due to the continuity of the Shannon entropy function. It is enough to prove the assertion for $N(w,X_1^n)/n$, since this implies the required claim for $N(w,X_1^n)/(n-\lvert w \rvert +1)$. From Lemma \ref{lem:koksmaapproximation}, we get that,

\begin{align*}
 \left\lvert \frac{N(w,X_1^n)}{n} - \frac{1}{b^{\lvert w \rvert}} \right\rvert < \delta  + \sum\limits_{\substack{t =-\infty\\ t\neq 0}}^{\infty} \frac{4}{t^2 \delta} \left\lvert\frac{\sum\limits_{j=1}^{n}e(tb^{(j-1)}x)}{n} \right\rvert
\end{align*}

Fix $\delta=\epsilon/2$. Now, 
\begin{align*}
 \left\lvert \frac{N(w,X_1^n)}{n} - \frac{1}{b^{\lvert w \rvert}} \right\rvert < \frac{\epsilon}{2}  + \sum\limits_{\substack{t =-\infty\\ t\neq 0}}^{\infty} \frac{8}{t^2 \epsilon} \left\lvert\frac{\sum\limits_{j=1}^{n}e(tb^{(j-1)}x)}{n} \right\rvert
\end{align*}
Since $\sum\limits_{i=t}^{\infty} \frac{1}{i^2} \leq \frac{1}{t}$, if $t \geq T'(\epsilon)= 64/\epsilon^2$, we \emph{cut off} an $\epsilon/4$ tail from the above sum to obtain the following.
\begin{align*}
 \left\lvert \frac{N(w,X_1^n)}{n} - \frac{1}{b^{\lvert w \rvert}} \right\rvert < \frac{\epsilon}{2}  + \frac{\epsilon}{4}+ \sum\limits_{\substack{t =-T'(\epsilon)\\ t\neq 0}}^{T'(\epsilon)} \frac{8}{t^2 \epsilon} \left\lvert\frac{\sum\limits_{j=1}^{n}e(tb^{(j-1)}x)}{n} \right\rvert
\end{align*}
If the Weyl averages for every non-zero parameter $t$ between $-T'(\epsilon)$ and $T'(\epsilon)$ are less than $\gamma'(\epsilon)=\epsilon^2/32\cdot T'(\epsilon)$, then the last term is less than $\epsilon/4$ and we get,
\begin{align}
\label{eq:smallweylaveragesimpliesnormality}
  \left\lvert \frac{N(w,X_1^n)}{n} - \frac{1}{b^{\lvert w \rvert}} \right\rvert < \frac{\epsilon}{2}  + \frac{\epsilon}{4}+ \frac{\epsilon}{4} = \epsilon.
\end{align}
The proof of the required claim is thus complete. The constants $T(\epsilon,b,k)$ and $\gamma(\epsilon,b,k)$ are obtained from the constants $T'$ and $\gamma'$ using the continuity of the Shannon entropy function (\cite{Khinchin57}, \cite{CovTho91}) in a straightforward manner.
\end{proof}

As in section \ref{sec:overview}, let $P_k^2$ denote the denote the index of the last step in the second substage of stage $k$. We make $P_k^2$ large
enough so that the following conditions are satisfied at the end of
the substage:

\begin{enumerate}
\label{text:secondsubstageconditions}
\item\label{item:secondsubstagecond1} (Entropy rates in base $v(k)$
  are close to 1)
  $\lvert H_l^{v(k)}(X(k)_1^n) -1 \rvert \leq 2^{-(k+1)}$ for every $l
  \leq k$ when $n=F_k^2$.
\item\label{item:secondsubstagecond2} (Exponential
  averages for base $v(k)$ are small) For every $t$
  with $\lvert t \rvert \leq T(2^{-(k+1)},v(k),k)$,
\begin{align}
\label{eq:gammabound1}
\left\lvert (F_k^2)^{-1} \sum\nolimits_{i=1}^{F_k^2}
e(tv(k)^{(j-1)}\xi) \right\rvert < 	\gamma(2^{-(k+1)},v(k),k)/2. 
\end{align}
where $T$ and $\gamma$ are the constants from Lemma
\ref{lem:smallweylaverageimpliesentropy1}. 
\item\label{item:secondsubstagecond3} $P_k^2 \geq
  \max\{M(\gamma(2^{-(k+1)},v(k),k)/2,v(k)),T(2^{-(k+1)},v(k),k)
  \}$, where $M$ is the constant from Lemma
  \ref{lem:exponentialtermsstrongbound}. 
\item\label{item:secondsubstagecond4} (Entropy rates in non-equivalent
  bases are close to 1) If there exists $k'<k$ such that
  $v(k')=v(k+1)$, then $\lvert H_l^{v(k+1)}(X(k+1)_1^n) -1 \rvert \leq
  2^{-k}$ for every $l \leq k$ when $n=\langle P_k^2+1;v(k+1)
  \rangle$.
 \item\label{item:secondsubstagecond5} For every $m \geq P_k^2$,
	 \begin{align}
\label{eq:kthtok+1thstagetransition}
b_m-a_m \geq \frac{
  L_{v(k+1),v^*(k+1)}(k,\delta_{k+1}(2^{-k})/2)+2k}{\min\{\delta_{k}(2^{-k}),\delta_{k+1}(2^{-k}),2^{-k}\}/2}+k 
\end{align} 
\item\label{item:secondsubstagecond6} The sequence $\langle u'(m)
  \rangle_{m=1}^{\infty}$ defined such that $u'(m)=u(m)$ for every $m
  \leq P_k^2$ and $u'(m)=v(k+1)$ for $m \geq P_k^2+1$, is a good
  sequence.  
\end{enumerate}
The constant $L_{v(k+1),v^*(k+1)}(k,\delta_{k+1}(2^{-k})/2)$ in condition \ref{item:secondsubstagecond5} is from Corollary \ref{cor:lowdiscrepancycombinationnormality}. Condition \ref{item:secondsubstagecond1} is satisfied for large enough $P_k^2$ because the occurrence probability of any finite string $w \in \Sigma^*_{v(k)}$ converges to $v(k)^{-\lvert w \rvert}$ on choosing $\xi_m$ from $\sigma_{m}(\xi_{m-1})$ according to Criterion \ref{item:criterion2} for sufficiently large number of $m$'s, as a consequence of Lemma \ref{lem:lowdiscrepancycombinationnormality}. For any $t$, on extending the second substage by increasing $P_k^2$, the corresponding exponential averages in (\ref{eq:gammabound1}) converges to $0$ as a consequence of the Weyl Criterion for normality (see \cite{Weyl1916},\cite{KuipersNiederreiterUniform}) and Lemma \ref{lem:lowdiscrepancycombinationnormality}. Therefore, condition \ref{item:secondsubstagecond2} is satisfied for large enough values of $P_k^2$.  Since $b_m-a_m \geq (e^{\sqrt{m+1}}-e^{\sqrt{m}}+m^2)/\log(v(k))-1$ and the right hand side of (\ref{eq:kthtok+1thstagetransition}) is a constant depending only on $k$, condition \ref{item:secondsubstagecond5} is satisfied for all sufficiently large $m$. 

It is easily verified from the definition of a good sequence that for large enough $P_k^2$, on setting $u(P_k+1)=v(k+1)$ the sequence $\langle u(m) \rangle_{m=1}^{P_k^2+1}$ satisfies all the conditions in the definition of good sequences (Definition \ref{def:goodsequence}). On extending the sequence from this value of $P_k^2+1$ onwards, by setting $u(m)=v(k+1)$ for every $k \geq P_k^2+2$, none of the conditions in Definition \ref{def:goodsequence} are violated. Therefore, condition \ref{item:secondsubstagecond6} is satisfied for all large enough values of $P_k^2$. During stage $1$, all the conditions in the definition of a good sequence are trivially satisfied. Therefore, the validity of condition \ref{item:secondsubstagecond6} at the end of every second substage, inductively ensures that the constructed sequence $\langle u(m) \rangle$ is a good sequence.

Proving that condition \ref{item:secondsubstagecond4} holds for large enough values of $P_k^2$, requires an argument using Lemmas \ref{lem:exponentialtermsstrongbound} and \ref{lem:smallweylaverageimpliesentropy1}.

\begin{lemma}
\label{lem:entropiesinvk+1convergesto1}
Condition \ref{item:secondsubstagecond4} in the construction is true
for all large enough values of $P_k^2$.
\end{lemma}  
\begin{proof}
	Recall that $P_k^1$ denotes the final value of $m$ that is set during the first substage of stage $k$. We know that $\langle u(m) \rangle_{m=1}^{P_k^1}$ satisfies the properties in the definition of a good sequence. On extending it further by setting $u(m)=v(k)$, all the properties in Definition \ref{def:goodsequence} remains satisfied. Therefore, if $u(m)$ is set during the second substage of stage $k$, the corresponding $\beta_m$ is more than $\beta_1/\sqrt[4]{m}$. From Definition \ref{def:goodsequence}, we also get that $\beta_m=\beta_{P_k^1}$ for any $m$ set during the second substage of stage $k$. Define the sequence $\langle c_i \rangle_{i=1}^{\infty}$ such that $c_i=\beta'_i$ if $i \leq P_k^1$ and $c_i=\beta'_{P_k^1}$ for $i \geq P_k^1+1$. It is easy to see that $\langle c_i \rangle$ is a non-increasing sequence such that $c_1=1/4$ and $c_i \geq c_1/\sqrt[4]{i}$. 
	
	From the statement of condition \ref{item:secondsubstagecond4}, we have $v(k')=v(k+1)$ for some $k'<k$. No consecutive elements are equal in the sequence $\langle r_k \rangle$. Furthermore, every equivalence class of numbers have a unique representative in $\langle r_k \rangle$. Hence, we get that $v(k') \not\sim v(k)$. Consider any index $m$ of $\langle u(m) \rangle$ that is set during the second substage of stage $k$. Since $\xi_m$ was chosen from $\sigma_m(\xi_{m-1})$ using Criterion \ref{item:criterion2}, from Lemma \ref{lem:ambound}, we get that,
	
	\begin{align*}
		\sum\limits_{\substack{t=-m\\t \neq 0}}^{m} \sum\limits_{\substack{h=1 \\ u(h) \not\sim u(m)}}^{m}  \left\lvert \sum\limits_{j=\langle m; u(h) \rangle+1}^{\langle m+1 ; u(h) \rangle} e(u(h)^{j-1} t \xi) \right\rvert^2 \leq \delta m^2 (\langle m+1 \rangle - \langle m \rangle)^{2-\beta_m}
	\end{align*}
	for every $t$ with $\lvert t \rvert \leq m$. Since $k'<k$ and $v(k+1)=v(k') \not\sim v(k)$, for any non-zero $t$ with $\lvert t \rvert < m$,
	\begin{align*}
		 \left\lvert \sum\limits_{j=\langle m; v(k+1) \rangle+1}^{\langle m+1 ; v(k+1) \rangle} e(v(k+1)^{j-1} t \xi) \right\rvert &\leq \delta m (\langle m+1 \rangle - \langle m \rangle)^{1-\beta'_m} \\
		 &=\delta m (\langle m+1 \rangle - \langle m \rangle)^{1-c_m}.
	\end{align*}
	Therefore, for any value of $m$ that is set during the second substage and any $l \leq \langle m+1 ; v(k+1) \rangle-\langle m ; v(k+1) \rangle$, 
	\begin{align*}
		\frac{1}{\langle m ; v(k+1) \rangle+l}&\sum\limits_{j=1}^{\langle m ; v(k+1) \rangle+l} e(v(k+1)^{j-1} t \xi) \\
		&\leq \frac{\langle P_k^1+1; v(k+1) \rangle}{\langle m ; v(k+1) \rangle+l} + \frac{\delta m \sum\limits_{i=P_k^1+1}^{m-1} \left( \langle i+1 \rangle -\langle i\rangle \right)^{1-c_i} +l}{\langle m ; v(k+1) \rangle+l}
	\end{align*}
	As $m$ goes to $\infty$,  the first term above goes to $0$. As a consequence of Lemma \ref{lem:exponentialtermsstrongbound} we also get that the second term goes to $0$ as $m \to \infty$. Let $P_k^2$ be large enough so that  
	\begin{align*}
	 \left\lvert \frac{1}{\langle  P_k^2+1; v(k+1) \rangle} \sum_{i=1}^{\langle  P_k^2+1; v(k+1) \rangle} e(tv(k+1)^{(j-1)}x) \right\rvert < 	\gamma(2^{-k},v(k+1),k)	
	\end{align*}
	for every $t$ with $\lvert t \rvert \leq T(2^{-k},v(k+1),k)$. Then, it follows from Lemma \ref{lem:smallweylaverageimpliesentropy1} that,
	\begin{align*}
	\left\lvert H_l^{v(k+1)}(X(k+1)_1^{\langle P_k^2+1;v(k+1) \rangle}) -1 \right\rvert \leq 2^{-k}  	
	\end{align*}
	for every $l \leq k$. The proof of the lemma is thus complete.
\end{proof}

\section{Verification}
\label{sec:verification}
In this section we prove that all the requirements given in section
\ref{sec:overview} are satisfied by the construction in section
\ref{sec:mainconstruction}. The following proofs along with the
argument provided at the end of section \ref{sec:overview} complete
the proof of Theorem \ref{thm:maintheorem}.

%
%

\begin{lemma}
\label{lem:fksk1sk2aresatisfied}
For all $k \geq 1$, the requirements $\mathcal{F}_k$,
$\mathcal{S}_{k,1}$ and $\mathcal{S}_{k,2}$ are met by the
construction.
\end{lemma}
\begin{proof}
$\mathcal{F}_k$ follows from the validity of condition
  \ref{item:firstsubstagecond1} from section
  \ref{subsec:firstsubstage} at the end of the first substage of every
  stage $k$.  From the validity of condition
  \ref{item:secondsubstagecond1} from section
  \ref{subsec:secondsubstage} at the end of the second substage of
  every stage $k$, we get that $\lvert
  H_l^{v(k)}(X(k)_1^n) -1 \rvert \leq 2^{-(k+1)}$ for every $l \leq k$
  when $n=F_k^2$. Therefore, $\mathcal{S}_{k,1}$ is satisfied for
  every $k \geq 1$.  $\mathcal{S}_{k,2}$ follows directly from the
  validity of condition \ref{item:secondsubstagecond4} from section
  \ref{subsec:secondsubstage} at the end of the second substage of
  every stage $k$.
\end{proof}
\begin{lemma}
\label{lem:rkissatisfied}
For every $k > 1$, the requirement $\mathcal{R}_k$ is satisfied by the
construction.
\end{lemma}
\begin{proof}
	Let $k'$ be any stage number below $k$ such that $v(k') \not\sim v(k)$. Without loss of generality, let us assume that there does not exist any $k''$ between $k'$ and $k$ such that $v(k'')=v(k')$. Since the equivalence class of base $v(k')$ has a unique representative in the sequence $\langle r_k \rangle$, we also get that $v(k'') \not\sim v(k')$ for any $k''$ between $k'$ and $k$. Since condition \ref{item:secondsubstagecond2} from section \ref{subsec:secondsubstage} is valid at the end of the second substage of stage $k'$, we have,
	\begin{align}
	\label{eq:rkproofeq1}
\left\lvert \frac{1}{F_{k'}^2} \sum\limits_{i=1}^{F_{k'}^2} e(tv(k')^{(j-1)}\xi) \right\rvert < 	\frac{\gamma(2^{-(k'+1)},v(k'),k')}{2}
\end{align}
for every $t$ with $\lvert t \rvert \leq T(2^{-(k'+1)},v(k'),k')$. From the validity of condition \ref{item:secondsubstagecond3} from section \ref{subsec:secondsubstage} at the end of the second substage of stage $k'$, we have,
\begin{align*}
	P_{k'}^2 \geq \max\{M(\gamma(2^{-(k'+1)},v(k'),k')/2,v(k')),T(2^{-(k'+1)},v(k'),k')\}.
\end{align*}
Let $m$ denote any index such that $P_{k'}^2+1 \leq m \leq P_{k}^2$. From Lemma \ref{lem:ambound}, we get that,
\begin{align*}
		\sum\limits_{\substack{t=-m\\t \neq 0}}^{m} \sum\limits_{\substack{h=1 \\ u(h) \not\sim u(m)}}^{m}  \left\lvert \sum\limits_{j=\langle m; u(h) \rangle+1}^{\langle m+1 ; u(h) \rangle} e(u(h)^{j-1} t \xi) \right\rvert^2 \leq \delta m^2 (\langle m+1 \rangle - \langle m \rangle)^{2-\beta_m}
	\end{align*}
for every $t$ with $\lvert t \rvert \leq m$. Since for every $k'' \in [k'+1,k]$, $v(k'') \not\sim v(k')$, from the above we obtain,
	\begin{align*}
		 \left\lvert \sum\limits_{j=\langle m; v(k') \rangle+1}^{\langle m+1 ; v(k') \rangle} e(v(k')^{j-1} t \xi) \right\rvert \leq \delta m (\langle m+1 \rangle - \langle m \rangle)^{1-\beta'_m}.
	\end{align*}
Define the sequence $\langle c_i \rangle_{i=1}^{\infty}$ as $c_i=\beta'_i$. It is easily verified that $\langle c_i \rangle$ is a non-increasing sequence satisfying  $c_1=1/4$ and $c_i \geq c_1 / \sqrt[4]{i}$ for every $i \geq 1$. Consider any $m$ satisfying $P_{k'}^2+1 \leq m \leq P_{k}^2$ and $l \leq \langle m+1 ; v(k') \rangle-\langle m ; v(k') \rangle$. Now, for any $t$ with $\lvert t \rvert \leq T(2^{-(k'+1)},v(k'),k') \leq m$, we have,
\begin{align*}
	& \frac{1}{\langle m ; v(k') \rangle+l}\left\lvert \sum\limits_{j=1}^{\langle m ; v(k') \rangle+l} e(v(k')^{j-1} t \xi) \right\rvert\\
	 &= \frac{1}{\langle m ; v(k') \rangle+l} \left\lvert \sum\limits_{i=1}^{\langle P_{k'}^2; v(k') \rangle} e(tv(k')^{(j-1)}\xi) \right\rvert \\
	 &+ \frac{1}{\langle m ; v(k') \rangle+l}\left\lvert\sum\limits_{j=\langle P_{k'}^2; v(k') \rangle+1}^{\langle m ; v(k') \rangle+l} e(v(k')^{j-1} t \xi) \right\rvert\\
	 &= \frac{1}{\langle m ; v(k') \rangle+l} \left\lvert\sum\limits_{i=1}^{F_{k'}^2} e(tv(k')^{(j-1)}\xi) \right\rvert + \frac{1}{\langle m ; v(k') \rangle+l}\left\lvert \sum\limits_{j=\langle P_{k'}^2; v(k') \rangle+1}^{\langle m ; v(k') \rangle+l} e(v(k')^{j-1} t \xi) \right\rvert\\
	 &\leq \left\lvert \frac{1}{F_{k'}^2} \sum\limits_{i=1}^{F_{k'}^2} e(tv(k')^{(j-1)}\xi) \right\rvert + \frac{\sum\limits_{i=P_{k'}^2}^{ m-1} \left( \langle i+1 \rangle - \langle i \rangle \right)^{1-\beta'_i} +l}{\langle m ; v(k') \rangle+l}\\
	  &\leq \frac{\gamma(2^{-(k'+1)},v(k'),k')}{2} + \frac{\sum\limits_{i=1}^{ m-1} \left( \langle i+1 \rangle - \langle i \rangle \right)^{1-c_i} +l}{\langle m ; v(k') \rangle+l}\\
	  &\leq \frac{\gamma(2^{-(k'+1)},v(k'),k')}{2} + \frac{\gamma(2^{-(k'+1)},v(k'),k')}{2} \\
	  &= \gamma(2^{-(k'+1)},v(k'),k').
\end{align*}
The second last inequality above follows from Lemma \ref{lem:exponentialtermsstrongbound} since,
\begin{align*}
 m \geq P_{k'}^2 \geq M(\gamma(2^{-(k'+1)},v(k'),k')/2,v(k')).
 \end{align*}

Finally, from the above inequalities and Lemma \ref{lem:smallweylaverageimpliesentropy1} we obtain that, $\lvert H_l^{v(k')}(X(k')_1^n) -1 \rvert \leq 2^{-(k'+1)}$ for every $l \leq k'$ when $\langle P_{k-1}^2+1; v(k') \rangle+1 \leq n \leq \langle P_{k}^2+1; v(k') \rangle$.

\end{proof}

\begin{lemma}
\label{lem:tk1issatisfied}
For every $k \geq 1$, the requirement $\mathcal{T}_{k,1}$ is met by
the construction.
\end{lemma}

\begin{proof}
Let $n \geq F_k^1$ denote the index of any digit that is fixed during the second substage of stage $k$ in the base $v(k)$ expansion of $\xi$. Now, for any $w \in \Sigma_{v(k)}^*$ let $\mathbb{P}^k(w,j_1,j_2)$ denote the fraction of $\lvert w \rvert$-length blocks containing $w$ among the digits in the base-$v(k)$ expansion of $\xi$ with indices in the range $j_1$ to $j_2$. Then,

\begin{align}
\label{eq:firstandsecondsubstageconvexcombination}
\mathbb{P}^{k}(w,1,n) &= \frac{F_k^1-\lvert w \rvert+1}{n-\lvert w \rvert+1} \mathbb{P}^k(w,1,F_k^1) + \frac{\lvert w \rvert-1}{n-\lvert w \rvert + 1} \mathbb{P}^k(w,F_k^1-\lvert w \rvert +2, F_k^1+\lvert w \rvert-1) \\
&+ \frac{n-F_k^1-\lvert w \rvert+1}{n-\lvert w \rvert+1} \mathbb{P}^k(w,F_k^1+1,n). \nonumber
\end{align}

For any string $w$ with length at most $k$ and any $n \geq F_k^1$, from (\ref{eq:firsttosecondsubstagetransition}) it follows that,

\begin{align*}
\frac{\lvert w \rvert-1}{n-\lvert w \rvert +1} \leq \frac{k}{n-k} \leq 	\frac{\min\{\delta_k(2^{-k}),2^{-k}\}}{2} \leq \frac{\delta_k(2^{-k})}{2}.
\end{align*}

For $n-F_k^1 \leq L_{v(k),v(k)}(k,\delta_k(2^{-k})/2)$, we have,
\begin{align*}
\frac{n-F_k^1-\lvert w \rvert+1}{n-\lvert w \rvert+1} &\leq \frac{L_{v(k),v(k)}(k,\delta_k(2^{-k})/2)}{n-\lvert w \rvert+1} \\
&\leq 	\frac{\min\{\delta_k(2^{-k}),2^{-k}\}}{2}\\
& \leq \frac{\delta_k(2^{-k})}{2}.
\end{align*}
Now, from the definition of $\delta_k(2^{-k})$ it follows that for any $l \leq k$,
\begin{align}
\label{eq:tk1issatisfiedproofeq1}
 \frac{1}{l \log v(k)} \left\lvert  H(\mathbb{P}^k_l(\cdot,1,n))-H(\mathbb{P}^k_l(\cdot,1,F_k^1)) \right\rvert &\leq \left\lvert  H(\mathbb{P}^k_l(\cdot,1,n))-H(\mathbb{P}^k_l(\cdot,1,F_k^1)) \right\rvert  \nonumber\\
 &\leq  \frac{1}{2^k} 
\end{align}
where $\mathbb{P}_l^k(\cdot,j_1,j_2)$ denotes the probability distribution over $\Sigma_{v(k)}^l$ such that for any $w \in \Sigma_{v(k)}^l$, $\mathbb{P}_l^k(w,j_1,j_2)$ is defined to be equal to $ \mathbb{P}^k(w,j_1,j_2)$. 

Since $\mathcal{F}_{k}$ is satisfied (as shown in Lemma \ref{lem:fksk1sk2aresatisfied}), we obtain,
\begin{align}
\label{eq:tk1issatisfiedproofeq2}
\left\lvert \frac{1}{l \log v(k)} H(\mathbb{P}^k_l(\cdot,1,F_k^1)) -q_{r_k} \right\rvert \leq \frac{1}{2^k}
\end{align}

for every $l \leq k$ when $n=F_k^1$. From (\ref{eq:tk1issatisfiedproofeq1}) and (\ref{eq:tk1issatisfiedproofeq2}) we get that among the first $L_{v(k),v(k)}(k,\delta_k(2^{-k})/2)$ digits fixed during the second substage, the $l$-length block entropy is inside $2^{-k}$ away from $q_{r_k}$ for every $l \leq k$.

Now, we consider the case when $n-F_k^1 > L_{v(k),v(k)}(k,\delta_k(2^{-k})/2)$. From Lemma \ref{lem:lowdiscrepancycombinationnormality}, it follows that when the aforementioned condition is satisfied, 

\begin{align*}
\left\lvert \mathbb{P}^k(w,F_k^1+1,n) - \frac{1}{v(k)^{\lvert w \rvert}} \right\rvert \leq \frac{\delta_k(2^{-k})	}{2}.
\end{align*}
for any $w \in \Sigma_{v(k)}^*$ with $\lvert w \rvert \leq k$. Now using the definition of $\delta_k(2^{-k})$ we obtain that,
\begin{align*}
\left\lvert \frac{1}{l \log v(k)} H(\mathbb{P}^k(\cdot,F_k^1+1,n)) - 1 \right\rvert \leq \frac{1}{2^k}. 
\end{align*}
for any $l \leq k$. Now, from the concavity of the Shannon entropy and (\ref{eq:firstandsecondsubstageconvexcombination}) we get that for any $l \leq k$,
\begin{align*}
\frac{1}{l\log v(k)} H(\mathbb{P}_l^k(\cdot,1,n)) &\geq \frac{F_k^1-l+1}{n-l+1} \times \frac{1}{l\log v(k)} H(\mathbb{P}^k(\cdot,1,F_k^1)) \\
&+ \frac{l-1}{n-l + 1} \times \frac{1}{l\log v(k)} H(\mathbb{P}^k(\cdot,F_k^1-l +2, F_k^1+l-1)) \\
&+ \frac{n-F_k^1-l+1}{n-l+1} \times \frac{1}{l\log v(k)}H(\mathbb{P}^k(\cdot,F_k^1+1,n))\\
&\geq \frac{F_k^1-l+1}{n-l+1} \times \frac{1}{l\log v(k)}H(\mathbb{P}^k(\cdot,1,F_k^1)) \\
&+ \frac{n-F_k^1-l+1}{n-l+1} \times \frac{1}{l\log v(k)}H(\mathbb{P}^k(\cdot,F_k^1+1,n)).
\end{align*}

Since $\mathcal{F}_{k}$ is satisfied (Lemma \ref{lem:fksk1sk2aresatisfied}), we know that,
\begin{align*}
\left\lvert \frac{1}{l \log v(k)} H(\mathbb{P}^k(\cdot,1,F_k^1)) -q_{r_k}  \right\rvert\leq \frac{1}{2^k}.	
\end{align*}
 Therefore it follows that for any $l \leq k$,
 \begin{align*}
 	 \frac{1}{l\log v(k)} H(\mathbb{P}^k(\cdot,1,n)) &\geq  \frac{F_k^1-l+1}{n-l+1} \times q_{r_k} + \frac{n-F_k^1-l+1}{n-l+1} \times 1 -\frac{1}{2^k}.	 
 \end{align*}

Setting,
\begin{align*}
\lambda_{l,k} = \frac{F_k^1-l+1}{n-l+1}	
\end{align*}
we get,

 \begin{align*}
 	 \frac{1}{l\log v(k)} H(\mathbb{P}^k(\cdot,1,n)) &\geq  \lambda_{l,k} \times q_{r_k} + (1-\lambda_{l,k}-\frac{l-1}{n-l+1}) \times 1 -\frac{1}{2^k} \\
 	 &= \lambda_{l,k} \times q_{r_k} + (1-\lambda_{l,k}) \times 1 -\frac{l-1}{n-l+1} -\frac{1}{2^k}
 \end{align*}
 
 Using the lower bound on the length of the first substage in (\ref{eq:firsttosecondsubstagetransition}), we get that,
 \begin{align*}
 \frac{l-1}{n-l+1} \leq \frac{k-1}{n-k+1} \leq \frac{k}{n-k} \leq \frac{1}{2^k}.
 \end{align*}

 Therefore,
\begin{align*}
	 \frac{1}{l\log v(k)} H(\mathbb{P}^k(\cdot,1,n)) &\geq \lambda_{l,k} \times q_{r_k} + (1-\lambda_{l,k}) \times 1 - \frac{1}{2^k}-\frac{1}{2^k}\\
	 &\geq q_{r_k} -  \frac{1}{2^{k-1}}.
\end{align*}

for every $l \leq k$ and $n \geq F_k^1 + L_{v(k),v(k)}(k,\delta_k(2^{-k})/2)$. Thus, we have shown that for every $l \leq k$ and $n \geq F_k^1$,
\begin{align*}
	 \frac{1}{l\log v(k)} H(\mathbb{P}^k(\cdot,1,n)) &\geq q_{r_k} -  \frac{1}{2^{k-1}}.
\end{align*}	

From the above condition we get that the requirement $\mathcal{T}_{k,1}$ is satisfied for every $k \geq 1$.	
\end{proof}

\begin{lemma}
\label{lem:tk2issatisfied}
For every $k \geq 1$, the requirement $\mathcal{T}_{k,2}$ is met by
the construction. 
\end{lemma}
\begin{proof}
We assume that there exists $k'<k$ such that $v(k')=v(k+1)$. Otherwise, $\mathcal{T}_{k,2}$ is vacuously satisfied. Let $F_k^2=\langle P_k^2 ; v(k+1) \rangle$ denote the index of the final digit fixed during the second substage of stage $k$ in the base $v(k+1)$ expansion of $\xi$. Let $n \geq F_k^2$ denote the index of any digit that is fixed during the first substage of stage $k+1$ in the base $v(k+1)$ expansion of $\xi$. As in the analysis of the transition between the first and second substages, for every $w \in \Sigma_s^*$, let $\mathbb{P}^{k+1}(w,j_1,j_2)$ denote the fraction of $\lvert w \rvert$-length blocks containing $w$ among the digits in the base-$v(k+1)$ expansion of $\xi$ with indices in the range $j_1$ to $j_2$. Then,

\begin{align}
\label{eq:kthandk+1thstageconvexcombination}
\mathbb{P}^{k+1}(w,1,n) &= \frac{F_k^2-\lvert w \rvert+1}{n-\lvert w \rvert+1} \mathbb{P}^{k+1}(w,1,F_k^2) \nonumber \\
&+ \frac{\lvert w \rvert-1}{n-\lvert w \rvert + 1} \mathbb{P}^{k+1}(w,F_k^2-\lvert w \rvert +2, F_k^2+\lvert w \rvert-1) \nonumber \\
&+ \frac{n-F_k^2-\lvert w \rvert+1}{n-\lvert w \rvert+1} \mathbb{P}^{k+1}(w,F_k^2+1,n).
\end{align}

For any string $w$ with length at most $k$ and any $n \geq F_k^2$, from (\ref{eq:kthtok+1thstagetransition}) it follows that,

\begin{align*}
\frac{\lvert w \rvert-1}{n-\lvert w \rvert +1} \leq \frac{k}{n-k} \leq 	\frac{\min\{\delta_{k+1}(2^{-k}),\delta_{k+1}(2^{-k}),2^{-k}\}}{2} \leq \frac{\delta_{k+1}(2^{-k})}{2}.
\end{align*}

For $n-F_k^2 \leq L_{v(k+1),v^*(k+1)}(k,\delta_{k+1}(2^{-k})/2)$, using inequality (\ref{eq:kthtok+1thstagetransition}) we get,

\begin{align*}
	\frac{n-F_k^2-\lvert w \rvert+1}{n-\lvert w \rvert+1} &\leq \frac{L_{v(k+1),v^*(k+1)}(k,\delta_{k+1}(2^{-k})/2)}{n-k}\\
	&\leq \frac{\min\{\delta_{k+1}(2^{-k}),\delta_{k+1}(2^{-k}),2^{-k}\}}{2}\\
	 &\leq \frac{\delta_{k+1}(2^{-k})}{2}.
\end{align*}

Now, from the definition of $\delta_{k+1}(2^{-k})$ it follows that for any $l \leq k$,
\begin{align}
\label{eq:tk2issatisfiedproofeq1}
 \frac{1}{l \log v(k+1)} \left\lvert  H(\mathbb{P}^{k+1}_l(\cdot,1,n))-H(\mathbb{P}^{k+1}_l(\cdot,1,F_k^2)) \right\rvert &\leq \left\lvert  H(\mathbb{P}^{k+1}_l(\cdot,1,n))-H(\mathbb{P}^{k+1}_l(\cdot,1,F_k^2)) \right\rvert  \nonumber\\
 &\leq  \frac{1}{2^k} 
\end{align}
where $\mathbb{P}_l^{k+1}(\cdot,j_1,j_2)$ denotes the probability distribution over $\Sigma_{v(k+1)}^l$ such that for any $w \in \Sigma_{v(k+1)}^l$, $\mathbb{P}_l^{k+1}(w,j_1,j_2)$ is defined to be equal to $ \mathbb{P}^{k+1}(w,j_1,j_2)$. 

Since there exists $k'<k$ such that $v(k')=v(k+1)$, $\mathcal{S}_{k,2}$ is satisfied (as shown in Lemma \ref{lem:fksk1sk2aresatisfied}). Therefore, we obtain,
\begin{align}
\label{eq:tk2issatisfiedproofeq2}
\left\lvert \frac{1}{l \log v(k+1)} H(\mathbb{P}^{k+1}_l(\cdot,1,F_k^2)) -q_{r_k} \right\rvert \leq \frac{1}{2^k}
\end{align}

for every $l \leq k$ when $n=F_k^2$. 

From (\ref{eq:tk2issatisfiedproofeq1}) and (\ref{eq:tk2issatisfiedproofeq2}) we get that among the first $L_{v(k+1),v^*(k+1)}(k,\delta_{k+1}(2^{-k})/2)$ digits fixed during the $(k+1)$\textsuperscript{th} stage, the $l$-length block entropy is inside $B_{2^{-k}}(1)$ for every $l \leq k$. 

Now, we consider the case when $n-F_k^2 > L_{v(k+1),v^*(k+1)}(k,\delta_{k+1}(2^{-k})/2)$. From Lemma \ref{lem:lowdiscrepancycombinationnormality}, it follows that when the aforementioned condition is satisfied, 

\begin{align*}
\left\lvert \mathbb{P}^{k+1}(w,F_k^2+1,n) - \frac{1}{v^*(k+1)^{\lvert w \rvert}} \right\rvert \leq \frac{\delta_{k+1}(2^{-k})	}{2}.
\end{align*}
for any $w \in \Sigma_{v^*(k+1)}^*$ with $\lvert w \rvert \leq k$. Now using the definition of $\delta_{k+1}(2^{-k})$ we obtain that,
\begin{align*}
&\left\lvert \frac{1}{l\log v(k+1)} H(\mathbb{P}_l^{k+1}(\cdot,F_k^2+1,n)) - q_{r_{k+1}} \right\rvert \\
&= \left\lvert \frac{1}{l\log v(k+1)} H(\mathbb{P}_l^{k+1}(\cdot,F_k^2+1,n)) - \frac{e_{r_{k+1}}}{d_{r_{k+1}}} \right\rvert\\ 
&=\left\lvert \frac{1}{l\log v(k+1)} H(\mathbb{P}_l^{k+1}(\cdot,F_k^2+1,n)) - \frac{\log v^*(k+1)}{\log v(k+1)} \right\rvert \\
&\leq \frac{1}{2^k}. 
\end{align*}
for any $l \leq k$. Now, from the concavity of the Shannon entropy and (\ref{eq:kthandk+1thstageconvexcombination}) we get that for any $l \leq k$,
\begin{align*}
\frac{1}{l\log v(k+1)} H(\mathbb{P}_l^{k+1}(\cdot,1,n)) &\geq \frac{F_k^2-l+1}{n-l+1}  \frac{1}{l\log v(k+1)} H(\mathbb{P}_l^{k+1}(\cdot,1,F_k^2)) \\
&+ \frac{l-1}{n-l + 1}\frac{1}{l\log v(k+1)} H(\mathbb{P}_l^{k+1}(\cdot,F_k^2-l +2, F_k^2+l-1)) \\
&+ \frac{n-F_k^2-l+1}{n-l+1} \frac{1}{l\log v(k+1)}H(\mathbb{P}_l^{k+1}(\cdot,F_k^2+1,n))\\
&\geq \frac{F_k^2-l+1}{n-l+1}  \frac{1}{l\log v(k+1)}H(\mathbb{P}_l^{k+1}(\cdot,1,F_k^2))\\
& + \frac{n-F_k^2-l+1}{n-l+1}  \frac{1}{l\log v(k+1)}H(\mathbb{P}_l^{k+1}(\cdot,F_k^2+1,n)).
\end{align*}

Since $\mathcal{S}_{k,2}$ is satisfied, we obtain,
\begin{align*}
\left\lvert \frac{1}{l \log v(k+1)} H(\mathbb{P}_l^{k+1}(\cdot,1,F_k^2)) -1  \right\rvert\leq \frac{1}{2^k}.	
\end{align*}
 Therefore it follows that for any $l \leq k$,
 \begin{align*}
 	 \frac{1}{l \log v(k+1)} H(\mathbb{P}_l^{k+1}(\cdot,1,n)) &\geq  \frac{F_k^2-l+1}{n-l+1} \times 1 + \frac{n-F_k^2-l+1}{n-l+1} \times q_{r_{k+1}} -\frac{1}{2^k}.	 
 \end{align*}

Setting,
\begin{align*}
\lambda_{l,k} = \frac{F_k^2-l+1}{n-l+1}	
\end{align*}
we get,

 \begin{align*}
 	 \frac{1}{l \log v(k+1)} H(\mathbb{P}_l^{k+1}(\cdot,1,n)) &\geq  \lambda_{l,k} \times 1 + (1-\lambda_{l,k}-\frac{l-1}{n-l+1}) \times q_{r_{k+1}} -\frac{1}{2^k} \\
 	 &= \lambda_{l,k} \times 1 + (1-\lambda_{l,k}) \times q_{r_{k+1}} -\frac{l-1}{n-l+1} -\frac{1}{2^k}
 \end{align*}
 
 Using the lower bound on the length of the first substage in (\ref{eq:kthtok+1thstagetransition}), we get that,
 \begin{align*}
 \frac{l-1}{n-l+1} \leq \frac{k-1}{n-k+1} \leq \frac{k}{n-k} \leq \frac{1}{2^k}.
 \end{align*}

 Therefore,
\begin{align*}
	 \frac{1}{l\log v(k+1)} H(\mathbb{P}_l^{k+1}(\cdot,1,n)) &\geq \lambda_{l,k} \times 1 + (1-\lambda_{l,k}) \times q_{r_{k+1}} - \frac{1}{2^k}-\frac{1}{2^k}\\
	 &\geq q_{r_{k+1}} -  \frac{1}{2^{k-1}}.
\end{align*}

The above holds for every $l \leq k$ and $n \geq F_k^2 + L_{v(k+1),v^*(k+1)}(k,\delta_{k+1}(2^{-k})/2)$. Thus, we have shown that for every $l \leq k$ and $n \geq F_k^2$,
\begin{align*}
	 \frac{1}{l\log v(k+1)} H(\mathbb{P}_l^{k+1}(\cdot,1,n)) &\geq q_{r_{k+1}} -  \frac{1}{2^{k-1}}.
\end{align*}
From the above condition we get that the requirement $\mathcal{T}_{k,2}$ is satisfied for every $k \geq 1$.	
\end{proof}

\section{Discussion and Open Questions}
It is open whether our main results are true in the setting of
finite-state strong dimension (\cite{athreya2007effective},
\cite{bourke2005entropy}). In particular: \textit{Does there exist an
  absolutely strong dimensioned number with finite-state strong
  dimension strictly between $0$ and $1$?}. The strong dimension of
$\xi$ is $1$. It is unclear how to modify our construction to bound
the limit superior of the block entropies away from $1$. This is
because while we control the block entropies in base $v(k)$ during
stage $k$, the block entropies in all bases with $k'<k$ and
$v(k')\not\sim v(k)$ are converging to $1$ (since the exponential
averages in these bases are \emph{getting smaller} after every
individual step within stage $k$). This behavior seems to be an
essential feature of constructions based on Schmidt's method
\cite{Schmidt62}. Hence, the question regarding \emph{absolute strong
dimension}, if such a number exists, may require new construction
techniques that are capable of stabilizing block entropies in
non-equivalent bases simultaneously around values strictly between $0$
and $1$.

\section{Table of notations and terminology}
\label{sec:tableofnotationsandterminology}

\begin{longtable}{| >{\centering\arraybackslash} m{4cm} | >{\centering\arraybackslash} m{6cm} | >{\centering\arraybackslash} m{6cm} |}
\hline
\textbf{Notation / Terminology} & \textbf{Meaning} & \textbf{Page of definition/description}\\
\hline 

$a_m$ & \emph{Starting index} of the range of digits fixed during step number $m$ in base $u(m)$ & Page \pageref{text:amdefinition} \\
\hline

$\alpha(r,s)$& Constant corresponding to bases $r$ and $s$ from Lemma \ref{lem:sinbound} & Page \pageref{lem:sinbound} \\
\hline

$A_m$ & The function involving exponential sums which is \emph{minimized} during step number $m$ of the construction & Page \pageref{eq:amxdefinition} \\
\hline

$b_m$ & \emph{Ending index} of the range of digits fixed during step number $m$ in base $u(m)$ & Page \pageref{text:bmdefinition} \\
\hline

$\beta_m$ & Constants defined for every $m$ using a given good sequence of natural numbers in Definition \ref{def:goodsequence} & Page \pageref{def:goodsequence}\\
\hline

$\beta'_m$ & Constants defined for every $m$ using a given good sequence of natural numbers as $\beta'_m=\beta_m/2$ & Page \pageref{text:betaidefinition}\\
\hline

$B_d(x)$ & The open neighborhood of radius $d$ around $x$ & Page \pageref{text:bdxdefinition} \\
\hline

$C_b$ & Constant from Theorem \ref{thm:lowdiscrepancytheorem} & Page \pageref{thm:lowdiscrepancytheorem}\\
\hline

$\chi_w $ & The characteristic function of string $w$ & Page \pageref{text:chiwdefinition}\\
\hline

\emph{``Criterion"}/\emph{``Criteria"}& The rules according to which $\xi_m$ is chosen in each individual step of the construction. They are defined in section \ref{subsec:schmidtsmethod}  & Page \pageref{text:criteria} \\
\hline

$d_b$ & The denominator of $q_b$ in the reduced form &  Page \pageref{text:ebdbdefinition}\\
\hline

$\delta$ & Constant which depends only on $u(1)$ from Lemma \ref{lem:ambound} & Page \pageref{lem:ambound}\\
\hline

$\delta_{k}(\epsilon)$ & Constant defined in section \ref{subsec:firstsubstage} & Page \pageref{text:deltakepsilondefinition} \\
\hline

$D^b_n(x)$ & The discrepancy function & Page \pageref{text:dbndefinition}\\
\hline

$\dim^b_{FS}(\xi)$ & The base-$b$ finite-state dimension of $\xi$ (Definition \ref{def:finitestatedimension}) & Page \pageref{def:finitestatedimension} \\
\hline

$e_b$  & The numerator of $q_b$ in the reduced form &  Page \pageref{text:ebdbdefinition}\\
\hline

$e(x)$ & The exponential function $e(x)=e(2\pi i x)$ & Page \pageref{text:exdefinition}\\
\hline

$F_k^1$ & The final digit in $X(k)$ fixed during the first substage of stage $k$ & Page \pageref{text:fk1definition} \\
\hline

$F_k^2$ & The final digit in $X(k)$ fixed during the second substage of stage $k$ & Page \pageref{text:fk2definition} \\
\hline

$\mathcal{F}_k$ & One of the end of substage requirements & Page \pageref{text:fkdefinition} \\
\hline

$\gamma(\epsilon,b,k)$ & Constant defined for every $\epsilon>0$, base $b$ and $k \in \N$ from Lemma \ref{lem:smallweylaverageimpliesentropy1}& Page \pageref{lem:smallweylaverageimpliesentropy1} \\
\hline

\emph{``Good sequence of natural numbers"} & See Definition \ref{def:goodsequence} & Page \pageref{def:goodsequence} \\
\hline

$H_l^b(w)$ & The $l$-length block entropy over $w$ (Definition \ref{def:finitestatedimension}) & Page \pageref{def:finitestatedimension}\\
\hline

$I_b(w)$ & Interval representing all numbers having base-$b$ expansion starting with $w$ & Page \pageref{text:ibwdefinition}\\
\hline

$I_k^1$ & The initial digit in $X(k)$ fixed during the first substage of stage $k$ & Page \pageref{text:ik1definition} \\
\hline

$I_k^2$ & The initial digit in $X(k)$ fixed during the second substage of stage $k$ & Page \pageref{text:ik2definition} \\
\hline

$L_{b,j}(k,\epsilon)$ &  Constant defined for every base $b$, $j \leq b$, $k \in \N$ and $\epsilon>0$ from Corollary \ref{cor:lowdiscrepancycombinationnormality} & Page \pageref{cor:lowdiscrepancycombinationnormality} \\
\hline

$L'_{b,j}(w,\epsilon)$& Constant defined for every base $b$, $j \leq b$, $w \in \Sigma_j^*$ and $\epsilon>0$ from Lemma \ref{lem:lowdiscrepancycombinationnormality} & Page \pageref{lem:lowdiscrepancycombinationnormality} \\
\hline

$\langle m \rangle$ & Quantity defined in section \ref{subsec:schmidtsmethod} & Page \pageref{text:anglemdefinition} \\
\hline

$\langle m; r \rangle$ & Quantity defined in section \ref{subsec:schmidtsmethod} & Page \pageref{text:anglemrdefinition} \\
\hline

$M(\epsilon,b)$ & Constant defined for every base $b$ and $\epsilon>0$ from Lemma \ref{lem:exponentialtermsstrongbound} & Page \pageref{lem:exponentialtermsstrongbound} \\
\hline

$N_b(\epsilon)$ & Constant from Lemma \ref{lem:lowdiscrepancyfinitestrings} & Page \pageref{lem:lowdiscrepancyfinitestrings} \\
\hline

$N(z,w)$ & The occurrence count of string $z$ in $w$ (Definition \ref{def:occurrencecountandprobability}) & Page \pageref{def:occurrencecountandprobability} \\
\hline

$p$ & Function from $p: \N \to \N$ used to identify the \emph{appropriate} sub-alphabets to be used during the first substages. The role of the sequence in the construction is described in section \ref{subsec:schmidtsmethod}. The exact function used in the proof of Theorem \ref{thm:maintheorem} is given in section \ref{sec:overview} & See Pages \pageref{text:pdefinition1} and \pageref{text:pdefinition2} \\
\hline

$P_k^1$ & The final step number in the first substage of stage $k$ & Page \pageref{text:pk1definition} \\
\hline

$P_k^2$ & The final step number in the second substage of stage $k$ & Page \pageref{text:pk2definition} \\
\hline

$P(z,w)$ & The occurrence probability of string $z$ in $w$ (Definition \ref{def:occurrencecountandprobability}) & Page \pageref{def:occurrencecountandprobability} \\
\hline

$\langle q_b \rangle_{b=1}^{\infty}$ & Sequence of rational dimensions from Theorem \ref{thm:maintheorem} & Page \pageref{thm:maintheorem} \\
\hline

$\langle r_k \rangle_{k=1}^{\infty}$ & Sequence of bases satisfying certain desirable properties chosen as the basis for the construction & Page \pageref{text:rkdefinition} \\
\hline

$R^b(x,n,\alpha_1,\alpha_2)$ & A measure of discrepancy with respect to the interval $(\alpha_1,\alpha_2)$ & Page \pageref{text:rbdefinition} \\
\hline

\emph{``Requirements"} & The six requirements $\mathcal{F}_k$, $\mathcal{S}_{k,1}$,  $\mathcal{S}_{k,2}$, $\mathcal{R}_k$,  $\mathcal{T}_{k,1}$,  $\mathcal{T}_{k,2}$, which if satisfied implies the proof of Theorem \ref{thm:maintheorem}. All the requirements are described in section \ref{sec:overview}. The proofs that the requirements are met are given in section \ref{sec:verification} & See Pages \pageref{text:requirementsbegin} (statements) and \pageref{sec:verification} (verification) \\
\hline

$\mathcal{R}_k$ & Requirement regarding the stability of non-equivalent base entropies & Page \pageref{text:requirementrkdefinition} \\
\hline

$\sigma_m(\lambda)$ & Set of numbers defined in section \ref{subsec:schmidtsmethod} & Page \pageref{text:sigmamlambdadefinition}\\
\hline

$\sigma^*_m(\lambda)$ & Set of numbers defined in section \ref{subsec:schmidtsmethod} & Page \pageref{text:sigmastarmlambdadefinition}\\
\hline

$\mathcal{S}_{k,1}$ & One of the end of substage requirements & Page \pageref{text:sk1definition} \\
\hline

$\mathcal{S}_{k,2}$ & One of the end of substage requirements & Page \pageref{text:sk2definition} \\
\hline

$\Sigma$ & Finite alphabet & Page \pageref{text:sigmadefinition}\\
\hline

$\Sigma_b$ & Base-$b$ alphabet & Page \pageref{text:sigmabdefinition} \\
\hline

\emph{``Stages"}& Each stage in our construction consists of consecutive stretches of \emph{steps}. An overview of the stages is given in section \ref{sec:overview} and the full details are given section \ref{sec:mainconstruction}  & See Pages \pageref{sec:overview} and \pageref{sec:mainconstruction} \\
\hline

\emph{``Steps"}& Denotes an \emph{individual step} from Schmidt's construction method from \cite{Schmidt62}. The details of an individual step are described in section \ref{subsec:schmidtsmethod}& Page \pageref{subsec:schmidtsmethod}\\
\hline

\emph{``Substages"}& Each stage in our construction is divided into two substages -the \emph{first substage} and the \emph{second substage}. An overview of the substages are given in section \ref{sec:overview}. See sections \ref{subsec:firstsubstage} and \ref{subsec:secondsubstage} for the full details of the substages. & Page \pageref{subsec:secondsubstage} \\
\hline

\emph{``Substage conditions"}& The conditions governing the lengths of the substages. The conditions for the first and second substage are given in section \ref{subsec:firstsubstage} and \ref{subsec:secondsubstage} respectively & See Pages  \pageref{text:firstsubstageconditions} and \pageref{text:secondsubstageconditions}\\
\hline

$T(\epsilon,b,k)$ & Constant defined for every $\epsilon>0$, base $b$ and $k \in \N$ from Lemma \ref{lem:smallweylaverageimpliesentropy1}& Page \pageref{lem:smallweylaverageimpliesentropy1} \\
\hline

$\mathcal{T}_{k,1}$ & One of the transition requirements & Page \pageref{text:tk1definition} \\
\hline

$\mathcal{T}_{k,2}$ & One of the transition requirements & Page \pageref{text:tk2definition} \\
\hline

$\langle u(m) \rangle_{m=1}^{\infty}$ & The sequence representing the bases in which digits are fixed at each individual step of the construction. The construction of the exact sequence used in Proof of Theorem \ref{thm:maintheorem} is  given in Section \ref{sec:mainconstruction} & Page \pageref{text:umdefinition} \\
\hline

$v(k)$ & The base in which digits are fixed during the $k$\textsuperscript{th} stage of the construction & Page \pageref{text:vkdefinition} \\
\hline

$v^*(k)$ & The \emph{sub-base} of $v(k)$ from whose elements are used to fix digits during the first substage of stage $k$ & Page \pageref{text:vstarkdefinition} \\
\hline

$v_b(w)$ & The rational number with base-$b$ expansion $w0^\infty$ & Page \pageref{text:vbwdefinition}\\
\hline

$w_1^n$ & Prefix containing the first $n$ bits of a finite string $w$ & Page \pageref{text:w1ndefinition}\\
\hline

$\xi_m$ & The rational number obtained at the end of step $m$ after fixing digits appropriately & Page \pageref{text:ximdefinition} \\
\hline

$X(k)$ & The infinite string representing the base-$v(k)$ expansion of the limit number $\xi$ & Page \pageref{text:xkdefinition} \\
\hline

$X_1^n$ & Prefix containing the first $n$ bits of an infinite string $X$ & Page \pageref{text:x1ndefinition}\\
\hline

\end{longtable}

\bibliographystyle{plainurl}
\bibliography{fair001,main,random}

\end{document}